\documentclass[manuscript,screen]{acmart}

\AtBeginDocument{%
  }

\setcopyright{acmlicensed}
\copyrightyear{2023}
\acmYear{2023}
\acmDOI{XXXXXXX.XXXXXXX}

\acmConference[LICS '23]{Logic in Computer Science}{June, 2023}{Boston, NY}
\acmISBN{978-1-4503-XXXX-X/18/06}

\usepackage{todonotes}
\usepackage{algorithm}
\usepackage[noend]{algpseudocode}
\usepackage{xspace}
\usepackage{pgf}
\usepackage{tikz}
\usepackage{xypic}
\usetikzlibrary{arrows,automata}
\usepackage{amsmath}
\usepackage{amsthm}
\usepackage{amsfonts}
\usepackage{comment}
\usepackage{mathtools}
\usepackage{xcolor}
\usepackage{mathrsfs}
\usepackage{tikz-cd} 

\usepackage{soul}

\usepackage{hyperref}
\usepackage{orcidlink}

\newenvironment{slremark}
    {\begin{remark}\rm}
    {\hfill $\triangleleft$ \end{remark}}

\newenvironment{slexample}
    {\begin{example}\rm}
    {\hfill $\triangleleft$ \end{example}}

\newcommand{\aasol}{almost-all-solution\xspace}

\newtheorem{theorem}{Theorem}    
\newtheorem{claim}[theorem]{Claim}
\newtheorem{remark}[theorem]{Remark}
\newtheorem{example}[theorem]{Example}

\newtheorem{corollary}[theorem]{Corollary}
\newtheorem{lemma}[theorem]{Lemma}
\theoremstyle{definition}
\newtheorem{question}[]{Question}
\newtheorem{definition}[theorem]{Definition}
\newtheorem{notation}[theorem]{Notation}

\newcommand{\polyring}[2]{#1[#2]}

\usepackage{amsmath}
\newcommand{\prettyexists}[2]{\exists #1 : #2}
\newcommand{\prettyforall}[2]{\forall #1 : #2}

\newcommand{\A}{\mathbb{A}}

\newcommand{\F}{\mathbb{F}}

\newcommand{\Rat}{\mathbb{Q}}
\newcommand{\R}{\mathbb{R}}

\newcommand{\Nat}{\mathbb{N}}
\newcommand{\Int}{\mathbb{Z}}
\newcommand{\Z}{\mathbb{Z}}

\newcommand{\conj}{\cup}

\newcommand{\supp}[1]{\textnormal{supp}(#1)}

\newcommand{\dom}[1]{\textnormal{dom}(#1)}

\newcommand{\prob}[3]{
\vspace{-4mm}
\emph{
\begin{description}
\item[{\sc #1:}]
\item[\bf Input: ] \ \ \ \ \ \  #2.
\item[\bf Question: ]  #3?
\end{description}
}
\vspace{1mm}
}

\newcommand{\compprobext}[3]{
\vspace{-2mm}
\emph{
\begin{description}
\item[{\sc #1:}]
\item[\bf Input: ] \ \ \ #2.
\item[\bf Output: ] #3.
\end{description}
}
\vspace{1mm}
}

\newcommand{\Lin}[1]{\text{\sc Fin-Lin}(#1)}
\newcommand{\GLin}[1]{\text{\sc Lin}(#1)}

\newcommand{\Span}[2]{\text{\sc Fin-Span}_{#1}\left(#2\right)}
\newcommand{\GSpan}[2]{\text{\sc Span}_{#1}\left(#2\right)}

\newcommand{\setof}[2]{\left\{\, #1 \; \middle| \; #2 \, \right\}}

\newcommand{\set}[1]{\left\{ #1 \right\}}
\newcommand{\para}[1]{\subsubsection*{\bf #1}}

\newcommand{\orbits}[1]{\text{\sc orbits}\left(#1\right)}
\newcommand{\Aut}[1]{\text{\sc Aut}_{#1}}
\newcommand{\size}[1]{|#1|}

\newcommand{\spread}[1]{\widetilde{#1}}

\newcommand{\innerprod}[2]{#1 \cdot #2 }
\newcommand{\subseteqfin}{\subseteq_\text{fin}}

\newcommand{\vr}[1]{\mathbf{#1}}
\newcommand{\otu}[2]{#1^{(#2)}}
\newcommand{\otun}[2]{\otu {#1} {#2}}
\newcommand{\utu}[2]{{#1 \choose #2}}
\newcommand{\constvr}[2]{{\vr {#1}}_{#2}}

\newcommand{\polyineqsolvname}{{\sc Poly-Ineq-Solv}\xspace}
\newcommand{\allpolyineqsolvname}{{\sc Almost-all-Poly-Ineq-Solv}\xspace}
\newcommand{\polyineqmaxname}{{\sc Poly-Ineq-Max}\xspace}
\newcommand{\allpolyineqmaxname}{{\sc Almost-all-Poly-Ineq-Max}\xspace}
\newcommand{\finname}{{\sc Fin-}}
\newcommand{\solvname}{{\sc Eq-Solv}}
\newcommand{\finsolvname}{{\sc Fin-Eq-Solv}}
\newcommand{\nonnegsolvname}{{\sc Nonneg-Eq-Solv}}
\newcommand{\nonnegmaxname}{{\sc Nonneg-Eq-Max}}
\newcommand{\finnonnegsolvname}{{\sc Fin-Nonneg-Eq-Solv}}
\newcommand{\ineqsolvname}{{\sc Ineq-Solv}}
\newcommand{\ineqmaxname}{{\sc Ineq-Max}}
\newcommand{\finineqsolvname}{{\sc Fin-Ineq-Solv}}
\newcommand{\finineqmaxname}{{\sc Fin-Ineq-Max}}

\newcommand{\mult}[2]{#1 \cdot #2}

\newcommand{\tofs}{\to_\text{fs}}
\newcommand{\tofin}{\to_\text{fin}}

\newcommand{\pair}[2]{(#1,#2)}

\newcommand{\tuple}[1]{\langle #1 \rangle}
\newcommand{\pow}[2]{\mathcal{P}_{#1}(#2)}

\newcommand{\trans}[1]{\stackrel{#1}{\longrightarrow}}

\newcommand{\setfromto}[2]{\set{#1, \ldots, #2}}
\newcommand{\setto}[1]{\setfromto 1 {#1}}

\newcommand{\proj}[2]{\Pi_{#1,#2}}
\newcommand{\orbval}[1]{\dot{#1}}
\newcommand{\orbsum}[1]{#1^{\Sigma}}

\renewcommand{\O}{U}
\newcommand{\x}{\chi}
\renewcommand{\a}{\alpha}
\renewcommand{\b}{\beta}
\newcommand{\g}{\gamma}
\newcommand{\s}{\sigma}
\renewcommand{\d}{\delta}
\renewcommand{\c}{\g}
\newcommand{\G}{\Gamma}
\newcommand{\ineqal}{\mathcal{E}}

\newcommand{\lipar}[2]{\textnormal{\textsc{hd}}_{#2}(#1)}
\newcommand{\li}[1]{\lipar {#1} {}}
\newcommand{\strictli}[1]{\lipar {#1} >}
\newcommand{\tl}[1]{\textsc{tl}(#1)}

\newcommand{\PTIME}{\textsc{PTime}\xspace}

\newcommand{\EXPTIME}{\textsc{ExpTime}\xspace}
\newcommand{\twoEXPTIME}{2\textsc{-ExpTime}\xspace}
\newcommand{\ackermann}{\textsc{Ackermann}\xspace}

\begin{document}

\title{Orbit-finite linear programming}


\author{Arka Ghosh}
\orcid{0000-0003-3839-8459}
\author{Piotr Hofman}
\orcid{0000-0001-9866-3723}
\author{S{\l}awomir Lasota}
\orcid{0000-0001-8674-4470}
\affiliation{%
  \institution{University of Warsaw}
  \country{Poland}
}

\renewcommand{\shortauthors}{Ghosh, Hofman, and Lasota}

\begin{abstract}
An infinite set is orbit-finite if, up to permutations of atoms, 
it has only finitely many elements. We study a generalisation of linear programming 
where constraints are expressed by an orbit-finite system of linear inequalities.
As our principal contribution we provide a decision procedure for checking if such a 
system has a real solution, and for computing the minimal/maximal value of a linear 
objective function over the solution set. We also show undecidability of these problems 
in case when only integer solutions are considered. Therefore orbit-finite linear 
programming is decidable, while orbit-finite integer linear programming is not.
\end{abstract}





\maketitle


\section{Introduction}
\label{sec:intro}
Applications of (integer) linear programming, and linear algebra in general, 
are ubiquitous in computer science (see e.g.~\cite{SilvaTC96,Cormenbook,Colcombet15}),
including recent and potential future applications to analysis of data-enriched 
models~\cite{HLT17,HL18,GSAH19,BKM21}.
%
Whenever finite (integer) linear programs arise in analysis of finite models of computation, 
   orbit-finite (integer) linear programs arise naturally in data-enriched versions of these models.
   For example, in decision problems for data Petri nets, such as reachability \cite{Nets-with-Tokens-which-Carry-Data} or
   continuous reachability \cite{GSAH19}; or in process mining \cite{WDHS08}.
Similar approach seems applicable also to
   structural properties or termination time of data Petri nets; and to  
   learning of probabilistic automata with registers.
   
This paper is a continuation of the study of \emph{orbit-finite} systems of linear equations~\cite{GHL22},
i.e., systems which are infinite but finite up to permutations.
In this setting one fixes a countably infinite set $\A$, 
whose elements are called \emph{atoms} (or data values)~\cite{atombook,Pitts:book}, 
assuming that atoms can only be accessed in a very limited way, namely can only be tested for equality.
Starting from atoms one builds a hierarchy of sets which are \emph{orbit-finite}: they are infinite, but finite up to permutations of atoms.
Along these lines, we study orbit-finite sets of linear inequalities, over an orbit-finite set of unknowns.



The main result of~\cite{GHL22} is a decision procedure to check if a given orbit-finite system
of equations is solvable.
This result is general and applies to solvability over a wide range of commutative rings, 
in particular to real and integer solvability.
%
%
%
In this paper we do a next step and extend the setting from equations to inequalities.
Our goal is algorithmic solvability of orbit-finite systems
of inequalities, but also optimisation 
of linear objective functions
over solution sets of such systems.
%
%
We call this problem \emph{orbit-finite (integer) linear programming}
(depending on whether the considered solutions are real or integer).

\begin{slexample} \label{ex:lp} 
For illustration, consider the set $\A$ as unknowns,
and the infinite system of constraints given by 
an infinite matrix whose rows and columns are indexed by $\A$:
\begin{align} \label{eq:matr}
& 
\begin{bmatrix}
\ 0 & 1 & 1 & \cdots \ \\
\ 1 & 0 & 1 & \cdots \ \\
\ 1 & 1 & 0 & \cdots \ \\
 \vdots     & \vdots & \vdots & \ddots \   
\end{bmatrix}
\cdot
\ \vr x 
\quad
\geq
\quad
\begin{bmatrix}
\, 1 \, \\
\, 1 \, \\
\, 1 \, \\
 \vdots 
\end{bmatrix}
\end{align}
Alternatively, one can write the infinite set of non-strict inequalities over unknowns $\a\in\A$, indexed by atoms $\b\in\A$:
\begin{align} \label{eq:constr}
\sum_{\a \in \A\setminus\set{\b}} \a \ \geq \ 1 \quad (\b \in \A).
\end{align}
Any permutation $\pi: \A \to \A$ induces a permutation of the inequalities by sending 
\[
\sum_{\a \in \A\setminus\set{\b}} \a \ \geq \ 1\qquad \stackrel{\pi}{\longmapsto}
\quad 
\sum_{\a \in \A\setminus\set{\pi(\b)}} \a  \ \geq \ 1,
\]
but the whole system~\eqref{eq:constr} is invariant under permutations of atoms.
Furthermore, up to permutations of atoms the system consists of just one equation -- it is one \emph{orbit};
in the sequel we consider orbit-finite systems (finite unions of orbits).
Likewise, the matrix $\A\times\A\to\R$ in \eqref{eq:matr} consists, up to permutation of atoms, of just two entries.
Indeed, its domain $\A\times\A$ is a union of two orbits: 
$\setof{\pair \a\b}{\a=\b}$ and $\setof{\pair \a \b}{\a\neq \b}$, and the matrix is constant
inside each orbit.
It is therefore invariant under permutations of atoms.

%

The system~\eqref{eq:constr} is solvable. 
For example, given $n>1$  atoms $S = \set{\a_1, \ldots, \a_n} \subseteq \A$, 
the vector $\vr{x}_n : \A \to \R$ defined by:
\[
\vr x_n(\a) = \frac 1 {n-1} \text{\ \ if \ }\a\in S, \qquad 
\vr x_n(\a) = 0 \text{\ \ if \ }\a\notin S,
\]
is a solution since the left-hand side of~\eqref{eq:constr} sums up to
$1$ if $\b\in S$, and to $\frac n {n-1}$ if $\b\notin S$.
\end{slexample}

%
Orbit-finite linear programming, i.e., 
optimisation of a linear objective function subject to an orbit-finite system of inequality constraints,
faces phenomena not present
in the classical setting.
For instance, as illustrated in the next example, 
the objective function may not achieve its optimum over solutions of non-strict inequalities.

\begin{slexample} \label{ex:lpmin} 
Suppose that we aim at \emph{minimization}, with respect to the constraints~\eqref{eq:constr},
of the value of the objective function:
\begin{align} \label{eq:objf}
S(\vr x) \ = \ 2\cdot \sum_{\a\in\A} \vr x(\a).
\end{align}
%
%
The function is invariant under permutations of atoms, and
its value is always greater than $2$.
Indeed, for every solution $\vr{x} : \A \to \R$ there is necessarily some $\b \in \A$ such that $\vr{x}(\b) > 0$,
and hence
\begin{align} \label{eq:gre1}
S(\vr x)  \ 
> \ 2 \cdot \!\! \sum_{\a \in \A\setminus\set{\b}} \vr{x}(\a) \ \geq \ 2.
\end{align}
What is the minimal value of the objective function? 
For solutions $\vr x_ n$ defined in Example~\ref{ex:lp}, the value $S(\vr x_n) = \frac {2n} {n-1}$
may be arbitrarily close to 2 but, according to \eqref{eq:gre1}, $S$ never achieves 2. 
Surprisingly, this is in contrast with classical linear programming where, whenever constraints are specified by non-strict inequalities and are solvable,
a linear objective function always achieves its minimum (or is unbounded from below).
%
\end{slexample}

Inequalities and unknowns can be indexed by more than one atom, as illustrated in the next example.
\begin{slexample}
\label{ex:Kirchoff}
Let $\otu \A 2 = \setof{\a\b\in\A^2}{a\neq\b}$ be the set of pairs of distinct atoms.
Consider a system whose inequalities and unknowns are indexed by
%
$\A \uplus \set{*}$ and
$\A  \uplus \otu \A 2$,
%
respectively.
Intuitively, unknowns correspond to vertices $\a$ and edges $\a\b$ of an infinite directed clique.
Let the system contain an inequality
\begin{align}  \label{eq:K1}
\sum_{\a \in \A} \a \ \geq \  1
\end{align}
enforcing the sum of values assigned to all vertices to be at least 1,
and the following inequalities
\begin{align}  \label{eq:K2}
\sum_{\b \in \A\setminus \set{\a}} \a\b \ - \ \a \ - \ \sum_{\b \in \A\setminus \set{\a}} \b\a  \ \geq \ 0 \qquad\quad (\a \in\A)
\end{align}
enforcing, for each vertex $\a\in\A$, the sum of values assigned to all outgoing edges
to be larger or equal to the sum of values assigned to all incoming edges, plus the value assigned
to the vertex $\a$.
In matrix form (0 entries of the matrix are omitted):

\vspace{-3mm}
\begin{align}
\label{eq:Kirchoff-matr}
\begin{aligned}
&\qquad\qquad\quad \begin{matrix} \A & \qquad\qquad \otu \A 2 \end{matrix}\\
&\begin{matrix}
\\
\A \\
\phantom{\ddots} \\
\set{*}
\end{matrix}
\left[
\begin{matrix}
  -1 & &  \\  
  & -1 &     \\     
  &  & {\ddots} \ \\ 
  \hline
   \ \;1 & \ \; 1 & {\cdots}  \ 
\end{matrix}
\right\rvert
\hspace{-1.25mm}
\left\rvert
\begin{matrix}
\ \\
\qquad \vr A  \qquad \\
\phantom{\ddots} \\
\hline    
\ 
\end{matrix}
\right]
\cdot
\ \vr x 
\quad
\geq
\quad
\begin{bmatrix}
\ 0 \ \\ \vdots \\ 0 \\ \hline 1
\end{bmatrix}
\end{aligned}
\end{align}
\vspace{1mm}

\noindent
where $\vr A$ is the oriented incidence matrix, namely for every distinct atoms $\a,\b\in\A$, 
\[
\vr A(\a, \a \b) = 1 \qquad
\vr A(\a, \b \a) = -1,
\]
and all other entries of $\vr A$ are 0.
As in previous examples, the system is invariant under permutations of atoms.
Solutions of the system correspond to directed graphs, 
whose vertices and edges are labeled
in accordance with constraints \eqref{eq:K1} and \eqref{eq:K2}.
We return to this system in Example \ref{ex:Kirchoff-cont} below, and 
in Sections \ref{sec:weak} and \ref{sec:max}.
\end{slexample}

The systems appearing in the above examples are invariant under all permutations of atoms.
As usual when working in \emph{sets with atoms} \cite{atombook} (also known as \emph{nominal sets} \cite{Pitts:book}),
we allow systems which, for some finite subset $S\subseteqfin \A$ (called a \emph{support}), 
are only invariant under permutations that fix $S$.
In the standard terminology of orbit-finite sets \cite{atombook}, we allow for \emph{finitely supported}
systems.
Likewise, 
we allow for finitely supported objective functions, and seek for finitely supported solutions. 
Equivalently, using terminology of sets with atoms,
we allow for \emph{orbit-finite} systems and objective functions, and seek for
orbit-finite solutions.
Note that solutions appearing in Examples \ref{ex:lp} and \ref{ex:lpmin} are \emph{finitary}, 
i.e. assign non-zero to only finitely
many unknowns, and are therefore finitely supported.
Each finite system is finitely supported, and thus orbit-finite linear programming is a generalisation
of the classical one.

\para{Contribution}

As our main contribution, we provide decision procedures for orbit-finite linear programming, both for
the decision problem of solvability of systems of inequalities, and for the optimisation problem.

The core ingredient of our approach is to reduce solvability (resp.~optimisation) of an orbit-finite system of inequalities to the analogous question
on a finite system which is \emph{polynomially parametrised}, i.e., where coefficients are 
univariate polynomials in an integer variable $n$.
The parameter $n$ corresponds, intuitively, to the number of atoms appearing in (a support of) a solution.
In this parametrised setting we ask for solvability for \emph{some} $n\in\Nat$, or for optimisation
when ranging over \emph{all} $n\in\Nat$.
We can compute an answer by encoding the problem into first-order
real arithmetic~\cite{Tar51,ra-book}, 
which yields decidability.
We provide also an efficient \PTIME algorithm for a relevant subclass of instances, 
resorting to polynomially many calls of 
classical linear programming.

\begin{slexample} 
\label{ex:poly}
For instance, the system \eqref{eq:matr} 
is transformed to the following two inequalities with one unknown $x$,
which are polynomially (actually, linearly) parametrised in a parameter $n$ (the details are exposed in Example \ref{ex:cont} in Section \ref{sec:red2poly}): 
\begin{align} \label{eq:ppar}
\begin{bmatrix}
n-1 \\
n
\end{bmatrix}
\cdot
x
\ \geq \ 
\begin{bmatrix}
n \\
n
\end{bmatrix}
\end{align}
The objective function $S$~\eqref{eq:objf}, on the other hand, is transformed to a 
non-parametrised linear map $x\mapsto 2\cdot x$.
For every $n > 1$, the system \eqref{eq:ppar} is solvable,
the minimal solution is $x=\frac{n}{n-1}$, and
the minimum of the objective function is $\frac{2n}{n-1}$.
Ranging over all $n\in\Nat$, the minimum can get arbitrarily close to 2, but never reaches 2.
\end{slexample}

Our reduction to a polynomially-parametrised linear program involves a size blowup
which is exponential in \emph{atom dimension}
(i.e., the maximal number of atoms appearing in an index of inequality or unknown of a system)
but polynomial in the number of its orbits.
In consequence, orbit-finite linear programming is solvable in \EXPTIME, and in \PTIME when
atom dimension is fixed.
In the setting of orbit-finite sets this means that the problem is feasible \cite{BT18}.
Therefore, in the latter case the complexity is not worse than in case of
classical linear programming.


\smallskip

One of cornerstones of the reduction is an observation that every system of inequalities that admits
a finitary solution, 
admits also a solution which is invariant under all permutations of its support, i.e., of atoms it uses.

\begin{slexample}
\label{ex:Kirchoff-cont}
As an illustration, let's prove that the system in Example \ref{ex:Kirchoff} has no finitary solution.
Indeed, if there existed a finite labeled  directed graph 
satisfying constraints \eqref{eq:K1} and \eqref{eq:K2}, 
by the above observation there would also exist a finite labeled directed
clique, where labels of all vertices are pairwise equal, and labels of all edges are pairwise equal as well.
In particular each vertex would carry the same value, 
necessarily positive due to constraint \eqref{eq:K1},
and
all edges incoming to a vertex would carry the same value as all outgoing edges.
These requirements are clearly contradictory with constraint \eqref{eq:K2}.
In conclusion, the system has no solutions.
\end{slexample}

As our second main result we prove undecidability of orbit-finite \emph{integer} linear programming,
already for the decision problem of solvability.
While the classical linear programming and integer linear programming are on the opposite sides of the feasibility border,
in case of orbit-finite systems the two problems are on the opposite sides of the decidability border.
One of key reasons behind undecidability of integer linear programming is that it can
express existence of a finite path. Specifically, if integer solutions are sought,
the system of inequalities of the form
\[
0 \ \leq \ \sum_{\a \in \A\setminus\set{\b}} \b \a \ \leq \ 1 \quad (\b \in \A)
\]
   allow us to say that every node $\beta$ has either no successors, or just one successor. 
   This clearly fails if real solutions are sought.
   
This article is an improved and extended version of the conference submission \cite{GHL23}.
The major improvement is lowering the complexity of orbit-finite linear programming from
\twoEXPTIME to \EXPTIME, and from \EXPTIME to \PTIME for fixed atom dimension.
This is achieved by identifying a suitable subclass of polynomially-parametrised 
linear programs which we show to be
solvable in \PTIME, and a reduction from orbit-finite linear programming to this subclass.

\para{Related research} 

This paper belongs to a wider research program that aims at lifting different aspects of 
theory of computation from finite to orbit-finite sets (essentially equivalent to first-order definable sets)
\cite{lics11,datamonoids,program-atoms,BKLT13,lmcs14,KLOT14,CL15,KKOT15,KLOT16,BT18,KKLO19}.

Our findings generalise, or are closely related to, some earlier 
results 
on systems of linear equations.
%
Systems studied in~\cite{HLT17} have row indexes of atom dimension 1.
In a more general but still restricted case studied in~\cite{HR21},
all row indexes are assumed to have the same atom dimension. 
%
%
Furthermore, columns of a matrix are assumed to be finitary in~\cite{HLT17,HR21}.
Both the papers investigate (nonnegative) integer solvability and are subsumed by
\cite{GHL22} (in terms of decidability, but not in terms of complexity), 
a starting point for our investigations.
Orbit-finite systems of equations are a special case of our present setting, as long as
   the solution domain is reals or integers. However, the setting of \cite{GHL22} is not 
   restricted to reals or integers, and allows for an arbitrary commutative ring as a solution domain
   (under some mild effectiveness assumption). This larger generality makes the results of 
   \cite{GHL22} and our results incomparable. Consequently, our methods are different than those
   of \cite{GHL22}.
   
The work~\cite{HL18} goes beyond \cite{HLT17} and investigates linear equations, in atom dimension 1,
over ordered atoms. 
Nonnegative integer solvability is decidable and equivalent to VAS reachability  
(and hence \ackermann-complete~\cite{CO,L,Lasota22}).


Systems in another related work~\cite{KKOT15} are over a finite field, contain only 
finite equations, and are studied as a special case of orbit-finite constraint satisfaction problems.
Furthermore, solutions sought are not restricted to be finitely-supported.

Orbit-finitely generated vector spaces were recently investigated in \cite{BKM21} and \cite{GHL22}.
The former paper shows that every chain of vector subspaces which are invariant under permutations of atoms eventually stabilises,
and apply this observation to prove decidability of zero-ness for orbit-finite \emph{weighted} automata.
The two papers jointly show that dual of an orbit-finitely generated vector space has an orbit-finite base.
In \cite{MinkowskiWeyl2021} the authors study cones in such spaces which are invariant under permutations of atoms, 
and extend accordingly theorems of Carath\'eodory and Minkowski-Weyl.

Finally, our technique discussed in Example \ref{ex:Kirchoff-cont}, and
developed formally in Section \ref{sec:weak},
seems to be reminiscent of (but independent from) 
the techniques in the recent work \cite{Asia}.


\para{Outline} 
After preliminaries on orbit-finite sets in Section \ref{sec:atoms},
in Section \ref{sec: linear eq} we introduce the setting of orbit-finite linear inequalities and
in Section \ref{sec:results} we state our results.
The rest of the paper is devoted to proofs.
In Sections \ref{sec:weak} and \ref{sec:poly}  we develop tools that are later used in 
decision procedures for linear programming in Sections \ref{sec:lp-decid-proof} and \ref{sec:max}.
Finally, Section \ref{sec:ilp-undecid-proof} contains the proof of undecidability of integer linear programming.
We conclude in Section \ref{sec:conc}.
Some routine or lengthy arguments are moved to Appendix.


\section{Preliminaries on orbit-finite sets} \label{sec:atoms}

%
Our definitions rely on basic notions and results of the theory of \emph{sets with atoms}~\cite{atombook}, 
also known as nominal sets~\cite{Pitts:book,lmcs14}. 
We only work with \emph{equality atoms} which have no additional structure except for
equality.

\para{Sets with atoms}

We fix a countably infinite set $\A$ 
whose elements we call \emph{atoms}.
Greek letters $\a,\b,\g,  \ldots$ are reserved to range over atoms.
%
The universe of sets with atoms is defined formally by a suitably adapted cumulative hierarchy of sets,
by transfinite induction: 
the only set of \emph{rank} 0 is the empty set; and for a cardinal $i$, a set of rank $i$ may contain, as elements, 
sets of rank smaller than $i$ as well as atoms.
In particular, nonempty subsets $X\subseteq\A$ have rank 1.

The group $\Aut{}$ of all permutations of $\A$, called in this paper \emph{atom automorphisms},
acts on sets with atoms by consistently renaming all atoms in a given set.
Formally, by another transfinite induction, for $\pi\in\Aut{}$ we define $\pi(X) = \setof{\pi(x)}{x\in X}$.
Via standard set-theoretic encodings of pairs or finite sequences we obtain, in particular, 
the pointwise action on pairs $\pi (xy)=\pi(x)\pi(y)$,
and likewise on finite sequences.  
Relations and functions from $X$ to $Y$ are considered as subsets of $X\times Y$.

We restrict to sets with atoms $X$
that only depend on finitely many atoms, in the following sense. 
For $T\subseteq \A$, let $\Aut T = \setof{\pi\in\Aut{}}{\pi(\a) = \a \text{ for every } \a \in T}$ be the set of atom automorphisms
that \emph{fix} $T$%
\footnote{
$\Aut T$ is often called the \emph{pointwise stabilizer} of $T$.
};
they are called $T$-automorphisms.
A finite set $T\subseteqfin\A$ (we use the symbol $\subseteqfin$ for finite subsets) is a \emph{support} of $X$ if for all
$\pi\in\Aut T$ it holds $\pi(X)=X$.
We also say: $T$ \emph{supports} $X$, or $X$ is \emph{$T$-supported}.
Thus a set is $T$-supported if and only if it is invariant under all $\pi \in \Aut T$.
As an important special case, a function $f:X\to Y$, 
understood as its diagram $\setof{(x,f(x))}{x\in X}$, 
is $T$-supported if  $f(\pi(x)) = \pi(f(x))$ for every argument $x$ and $\pi\in\Aut{T}$.
In particular, whenever $f$ is $T$-supported, its domain $X$ is necessarily $T$-supported too.
A $T$-supported set is also $T'$-supported, assuming $T\subseteq T'$.

A set $x$ is \emph{finitely supported} if it has some finite support; in this case
$x$ always has the least (inclusion-wise) support, denoted $\supp x$, called \emph{the support} of $x$ 
(cf.~\cite[Sect.~6]{atombook}).
Thus $x$ is $T$-supported if and only if $\supp x \subseteq T$.
Sets supported by $\emptyset$ (i.e., invariant under all atom automorphisms) we call \emph{equivariant}.

\begin{slexample}
Given $\a,\b\in\A$, the support of the set $\A\setminus\set{\a,\b}$ is $\set{\a,\b}$.
The set $\A^2$ and the projection function $\pi_1 : \A^2 \to \A : (\a, \b) \mapsto \a$
are both equivariant; 
and the support of a tuple $\tuple{\a_1, \ldots, \a_n}\in\A^n$, encoded as a set in a standard way, is
the set of atoms $\set{\a_1, \ldots, \a_n}$ appearing in it.
\end{slexample}

From now on, we shall only consider sets  that are hereditarily finitely
supported, 
i.e., ones that have a finite support, whose every element has some finite support,
and so on.


\para{Orbit-finite sets} 

Let $T\subseteqfin\A$. Two atoms or sets $x, y$ 
are \emph{in the same $T$-orbit} if $\pi(x) = y$ for some $\pi\in\Aut T$.
This equivalence relation splits atoms and sets 
into equivalence classes, which we call \emph{$T$-orbits};
$\emptyset$-orbits we call equivariant orbits, or simply \emph{orbits}.
By the very definition, every $T$-orbit $\O$ is $T$-supported: $\supp \O \subseteq T$.%
\footnote{The inclusion may be strict, for singleton $T$-orbits $O$.
For instance, the singleton $\set{\a}\subseteq\A$ is a $\set{\a}$-orbit, but also a $\set{\a,\b}$-orbit for $\b\neq \a$.
}

%
$T$-supported sets are exactly unions of (necessarily disjoint) $T$-orbits.
Finite unions of $T$-orbits, for any $T\subseteqfin\A$, are called  \emph{orbit-finite} sets.
Orbit-finiteness is stable under orbit-refinement: if $T\subseteq T' \subseteqfin \A$, a finite union of $T$-orbits is also a finite union
of $T'$-orbits (but the number of orbits may increase, cf.~\cite[Theorem~3.16]{atombook}). 

\begin{slexample}
Examples of orbit-finite sets are: 
\begin{itemize}
\item the set of atoms $\A$ (1 orbit); 
\item $\A \setminus \set{\a}$ for some $\a\in\A$ (1 $\set{\a}$-orbit);
\item pairs of atoms $\A^2$ (2 orbits: diagonal $\setof{\a\a}{\a\in\A}$ and off-diagonal $\otu \A 2 = \setof{\a\b\in\A^2}{\a\neq \b}$);
\item $n$-tuples of atoms $\A^n$ for $n\in \Nat$; each orbit $U\subseteq \A^n$ contains all $n$-tuples of the same equality type,
where by the \emph{equality type} of an $n$-tuple $a_1 \ldots a_n\in\A^n$ we mean
the set $\setof{(i,j)}{a_i = a_j}$;
\item non-repeating $n$-tuples of atoms 
$
\otu \A n = \setof{\!\a_1\ldots\a_n\in\A^n\!}{\a_i\neq\a_j \text{ for }i\neq j}
$  (1 orbit);
\item $n$-sets of atoms $\utu \A n = \setof{X\subseteq \A}{\size X = n}$ (1 orbit).
\end{itemize}
All of them are equivariant, except $\A\setminus\set{\a}$.
On the other hand,
the set  $\pow {\text{fin}} \A$ of all finite subsets of atoms is orbit-infinite as cardinality is an invariant of each orbit.
\end{slexample}

We now state few properties to be used in the sequel.
For $T\subseteqfin\A$,
each $T$-orbit $\O\subseteq \otu \A n$ is determined by fixing pairwise distinct atoms
from $T$ on a subset $I\subseteq \setto n$ of positions, while allowing arbitrary atoms from $\A\setminus T$ on
remaining positions $\setto n\setminus I$:
\begin{lemma} \label{lem:tupleorbit}
Let $T\subseteqfin \A$.
$T$-orbits $\O\subseteq \otu \A n$ are exactly sets of the form
\begin{align} \label{eq:Torb}
\begin{aligned}
\setof{a \in \otu \A n}{\proj n I(a) = u, \ \ \proj n {\setto n \setminus I} (a) \in \otu {(\A\setminus T)} {n-\ell}},
\end{aligned}
\end{align}
where $I \subseteq \setto n$, $\size I = \ell$, and $u\in\otu T \ell$.
The projection $\proj n I : \otu \A n \to \otu \A \ell$ is defined in the expected way.
\end{lemma}
%
Indeed, the set \eqref{eq:Torb} is invariant under all $T$-automorphisms,
and each two of its elements
are related by some $T$-automorphism.
%
%
%
The orbit \eqref{eq:Torb} is in $T$-supported bijection with $\otu {(\A\setminus T)} {n-\ell}$.
This is a special case of a general property of
every $T$-orbit, not necessarily included in $\otu \A n$.
The following lemma is proved exactly as \cite[Theorem 6.3]{atombook}
and provides finite representations of $T$-orbits:
\begin{lemma}\label{lem:repn}
Let $T \subseteqfin \A$.
Every $T$-orbit admits a $T$-supported bijection to a set of the form 
$
\otu {(\A \setminus T)} n/_G,
$
for some $n \in \Nat$ and some subgroup $G$ of $S_n$.
\end{lemma}
%


Recall that each orbit $U\subseteq \A^n$ contains all $n$-tuples of the same equality type.
In particular,  each orbit included in $\otu \A n \times \otu \A m \subseteq \A^{n+m}$
is induced by a partial injection $\iota$ from $\setto n$ to $\setto m$:
\begin{lemma} \label{lem:orbitpair}
Orbits $\O\subseteq \otu \A n \times \otu \A m$
are exactly sets of the form
$
\setof{ (a, b) \in \otu \A n \times \otu \A m
}{\prettyforall{i,j}{a(i) = b(j) \iff \iota(i) = j}},
$
where $\iota$ is a partial injection from $\setto n$ to $\setto m$.
\end{lemma}
%
%

Atom automorphisms preserve the size of the support: $\size{\supp X} = \size{\supp{\pi(X)}}$ for every set $X$ and
$\pi\in\Aut{}$.
We define \emph{atom dimension} of an orbit as the size of the support of its elements.
For instance, atom dimension of $\otu \A n$ is $n$.



\section{Orbit-finite (integer) linear programming} 
\label{sec:problems}
\label{sec: linear eq}




We introduce now the setting of linear inequalities we work with, and formulate our main results. 
We are working in vector spaces over the real%
\footnote{
All the results of the paper still hold if reals $\R$ are replaced by rationals $\Rat$
in all the subsequent definitions and results.
}
field $\R$,
where vectors are indexed by a fixed orbit-finite set $B$, i.e., are functions $\vr v : B\to\R$.
Observe that such a function $\vr v$, understood as its diagram $\setof{(b, \vr v(b))}{b\in B}$,
is orbit-finite exactly when it is finitely supported
(according to definitions in Section \ref{sec:atoms}).

\begin{definition}
By a \emph{vector} over $B$ we mean any orbit-finite (i.e., finitely-supported) function
$\vr v$ from $B$ to $\R$, 
written $\vr v : B\tofs\R$ (vectors are written using boldface). 
Vectors with integer range, $\vr v : B \tofs \Int$, we call \emph{integer} vectors.
\end{definition}

The set of all vectors over $B$ we denote by $\GLin B = B \tofs \R$. 
It is a vector space, 
with pointwise addition and scalar multiplication:
for $\vr v, \vr v'\in\GLin B$, $b\in B$ and $q\in \R$, we have
$(\vr v+ \vr v')(b) = \vr v(b) + \vr v'(b)$ and $(q \cdot \vr v)(b) = q\cdot \vr v(b)$.
These operation preserve the property of being finitely-supported, e.g., 
$\supp{\vr{v} + \vr{v'}} \subseteq \supp{\vr{v}} \cup \supp{\vr{v}'}$.
%
%
%
We define the \emph{domain} of a vector $\vr v\in\GLin B$ as $\dom {\vr v} = \setof{b\in B}{\vr v(b)\neq 0}$.
%
%
A vector $\vr v$ over $B$ is \emph{finitary}, written $\vr v : B\tofin \R$, 
if $\dom {\vr v}$ is finite, i.e., $\vr v(b) = 0$ for almost all  $b\in B$.

\begin{slexample}
\label{ex:vect}
Let $B = \otu \A 2$. Let $\a, \b\in \A$ be two fixed atoms.
The function $\vr v : B \to \R$ defined, for  $\x,\g\in\A\setminus\set{\a,\b}$, by
\begin{align*}
\vr v(\a \x) &= \vr v(\x \a) = -1 &
\vr v(\a \b) &= \vr v(\b \a) = 3  \\ 
\vr v(\b \x) &= \vr v(\x \b) = -2
&
\vr v(\x \g) &= 0
\end{align*}
is an $\set{\a,\b}$-supported integer vector over $B$.
It is not finitary, as $\dom {\vr v} = \setof{\d\s\in\otu \A 2}{\set{\d,\s}\cap\set{\a,\b} \neq \emptyset}$ is infinite.
Finitary $\set{\a,\b}$-supported vectors over $B$ assign $0$ to all elements of $B$ except for $\a\b$ and $\b\a$.
\end{slexample}

A finitary vector $\vr v$ with domain $\dom {\vr v} = \set{b_1, \ldots, b_k}$ such that
$\vr v(b_1) = q_1, \ldots, \vr v(b_k) = q_k$, 
may be identified with a formal linear combination of elements 
of $B$:
\begin{align} \label{eq:formalsum}
\vr v \ = \ q_1 \cdot b_1 + \ldots + q_k \cdot b_k.  
\end{align}
The subspace of $\GLin B$ consisting of all finitary vectors we denote by $\Lin B = B\tofin\R$.
For finite $B$ of size $\size B = n$, $\GLin B = \Lin B$ is isomorphic to $\R^{n}$.
%

For a subset $X\subseteq B$, we denote by $\constvr 1 X\in\GLin B$ the characteristic function of $X$,
i.e., the vector that maps each element of  $X$ to $1$ and all elements of $B\setminus X$ to $0$:
\[
\constvr 1 X : b \mapsto \begin{cases}
1 & \text{ if } b \in X \\
0 & \text{ otherwise.}
\end{cases} 
\]
We write $\constvr 1 b$ instead of $\constvr 1 {\set{b}}$, and $\vr 1$ instead of $\constvr 1 B$.


%


\begin{lemma} \label{lem:1O}
Let $T\subseteqfin \A$ and $\vr v \in \GLin {B}$ such that $\supp {\vr v}\subseteq T$. Then
\begin{enumerate}
\item[(i)] $\vr v$ is constant, when restricted to every $T$-orbit $\O\subseteq B$;
\item[(ii)] $\vr v$ is a linear combination of characteristic vectors $\constvr 1 \O$ of $T$-orbits $\O\subseteq B$.
\end{enumerate}
\end{lemma}
\begin{proof}
%
The first part follows immediately as $T$ supports ${\vr v}$.
As required in the second part, we have:
\begin{align} \label{eq:1O}
\vr v = \sum_\O \vr v(b_\O) \cdot \constvr 1 \O,
\end{align}
where $\O$ ranges over finitely many 
$T$-orbits $\O\subseteq B$, and 
$b_{\O}\in\O$ are arbitrarily chosen representatives of $T$-orbits. 
\end{proof}
%
%
\begin{notation} \label{not:orbval}
In the sequel, whenever we know that a vector $\vr v : B\tofs \R$ is constant over a $T$-orbit
$\O \subseteq B$,
we may write $\orbval{\vr v}(\O)$ instead of $\vr v(b)$, where $b\in\O$.
In particular, when $\vr v$ is equivariant, we have the \emph{orbit-value} vector
\[
\orbval{\vr v} : \orbits B \to \R,
\]
where $\orbits B$ stands for the set of all equivariant orbits $\O$ included in $B$.
\end{notation}


We note that the inner product of  vectors $\vr x, \vr y\in\GLin B$, 
\[
\innerprod {\vr x} {\vr y} \ = \ \sum_{b\in B} \vr x(b)\, \vr y(b),
\]
is not always well-defined.
We consider the right-hand side sum as well-defined when there are only finitely many $b\in B$
for which both $\vr x(b)$ and $\vr y(b)$ are non-zero
(equivalently, the intersection $\dom {\vr x} \cap \dom {\vr y}$ is finite).%
\footnote{In particular, $\innerprod {\vr x} {\vr y}$ is always well-defined when one of $\vr x, \vr y$ is finitary.}

%


\para{Orbit-finite systems of linear inequalities}

Fix an orbit-finite set $C$ (it can be thought of as the set of unknowns).
By a linear inequality over $C$ we mean a pair $e = (\vr a, t)$ where 
$\vr a : C \tofs \Int$ is an integer vector of left-hand side coefficients and $t\in\Z$ is a right-hand side target value%
\footnote{Rational coefficients and target are easily scaled up to integers.}.
An \emph{$\R$-solution} (real solution) of $e$ is any vector $\vr x : C\tofs \R$ such that the inner product
$\innerprod {\vr a} {\vr x}$ is well-defined and
\[
\innerprod {\vr a} {\vr x} \geq t;
\]
$\vr x$ is an \emph{$\Int$-solution} (integer solution) if $\vr x: C\tofs \Int$.
We may also consider constrained solutions, e.g., 
\emph{finitary} ones.
A linear \emph{equality} $\innerprod {\vr a} {\vr x} = t$ may be encoded by two opposite inequalities:
\[
\innerprod {\vr a} {\vr x} \geq t \qquad
- \innerprod {\vr a} {\vr x} \geq -t.
\] 

In this paper we investigate sets of inequalities indexed by an orbit-finite set.
Formally, an orbit-finite system of linear inequalities (over $C$)
is the pair $(\vr A, \vr t)$, where $\vr A : B \times C \tofs \Z$ is an integer \emph{matrix} 
with row index $B$ and column index $C$, and $\vr t : B \tofs \Z$ is an integer
\emph{target} vector:
%
\begin{align*}
& \quad\ \ \ \,  \begin{matrix} \ \cdots \quad &  c  & \quad \ \  \cdots\ \end{matrix} \\
& \begin{matrix} \vdots \\ b \\ \vdots \end{matrix} \ \ 
\begin{bmatrix}
\  & \vdots & \ \\
\ \cdots & \vr A(b,c) & \cdots\ \  \\
      & \vdots & \  
\end{bmatrix}
\qquad
\begin{bmatrix}
\vdots  \\
\vr t(b)\\
\vdots   
\end{bmatrix}
\end{align*}
%

\noindent
For $b\in B$ we denote by $\vr A(b, \_) \in \GLin C$ the corresponding (row) vector.
A solution of a system $(\vr A, \vr t)$ is any vector $\vr x\in\GLin C$ which is a solution of all inequalities 
$(\vr A(b, \_), \vr t(b))$, $b\in B$.
Equivalently, $\vr x$ is a solution if $\mult{\vr A}{\vr x} \geq \vr t$, where $\geq$ is the pointwise order on vectors,
and the (partial) operation of multiplication of a matrix $\vr A$ by a vector $\vr x$ is defined in an expected way:
\[
(\mult {\vr A}{\vr x})(b) = \innerprod {\vr A(b, \_)} {\vr x}
\]
for every $b\in B$.
The result $\mult {\vr A} {\vr x} \in \GLin B$ is well-defined if $\innerprod {\vr A(b, \_)}{\vr x}$ is well-defined for all $b\in B$.

By the following examples, restricting to equivariant, finitary or integer solutions only
has impact on solvability:

\begin{slexample}\label{ex:nofin1}
Let columns be indexed by $C = \A$,
and consider the system consisting of just one infinitary inequality
$(\constvr 1 {\A}, 1)$ ($B$ is thus a singleton).
Identifying column indexes $\a \in \A$ with unknowns, the inequality may be written as:
\[
\sum_{\a\in\A} \! \a \ \geq \ 1.
\]
The inequality has an integer (finitary) solution, i.e., $\vr x = \constvr 1 \a$ for any $\a\in\A$, but no 
equivariant one. 
Indeed, equivariant vectors $\vr x : \A \tofs \R$ are necessarily 
constant ones $\vr x = r\cdot \constvr 1 {\A}$ (cf.~Lemma~\ref{lem:1O}), and then
the inner product
\[
\innerprod {\constvr 1 {\A}} {\vr x} 
= \sum_{\a \in \A} \vr x(\a)
= \sum_{\a \in \A} r
\]
is well-defined only if $r=0$, i.e. $\vr x(\a) = 0$ for all $\a \in \A$.
\end{slexample}

\begin{slexample}
Let columns be indexed by $C = \otu \A 2$ and rows by $B = \utu \A 2$.
Consider the system containing, 
for every $\set{\a,\b}\in B$, the inequality
$(\constvr 1 {\a\b} + \constvr 1 {\b\a}, 1)$.
Using the formal-sum notation as in~\eqref{eq:formalsum} it may be written as
$(\a\b + \b\a, 1)$ or,
identifying column indexes $\a \b\in C$ with unknowns, as: 
\[
\a\b + \b\a \ \geq \ 1 \qquad\qquad (\a,\b\in\A, \a \neq \b).
\]
All the equations are thus finitary, and the target is $\vr t = \constvr 1 B$.
The constant vector $\vr x = \constvr {\frac 1 2} {} : \a\b \mapsto \frac 1 2$ is a solution,
even if we extend the system with symmetric inequalities 
\[
\a\b + \b\a \ \leq \ 1 \qquad\qquad (\a,\b\in\A, \a \neq \b).
\]
The extended system has no finitary solution.
It has no integer 
solution either.
Indeed, since we restrict to finitely supported solutions only, any such solution $\vr x$ 
necessarily satisfies, for every distinct atoms $\a,\b \in \A\setminus \supp{\vr x}$, 
the equality $\vr x(\a\b) = \vr x(\b\a)$, 
which is incompatible with $\vr x(\a\b) + \vr x(\b\a) = 1$.
\end{slexample}

\section{Results} \label{sec:results}

\para{Solvability problems}

We investigate decision problems of solvability of orbit-finite systems of inequalities over the ring of reals
or integers.
Consequently, we use $\F$ to stand either for $\R$ or $\Int$.
We identify a couple of variants. 
In the first one we ask about existence of a finitely-supported solution:

\prob{\ineqsolvname$(\F)$}
{An orbit-finite system of linear inequalities}
{Does it have an \xspace $\F$-solution}

\noindent
We recall that solutions are finitely-supported (equivalently, orbit-finite), by definition.
A closely related variant is solvability of equalities, under the restriction to nonnegative solutions only:

\prob{\nonnegsolvname$(\F)$}
{An orbit-finite system of linear equations}
{Does it have a nonnegative \xspace $\F$-solution}
\noindent
Furthermore, both the problems have "finitary" versions, where one seeks for finitary solutions only,
denoted as \finineqsolvname$(\F)$ and \finnonnegsolvname$(\F)$, respectively.

Three out of the four problems are inter-reducible, and hence equi-decidable,
both for $\F = \R$ and $\F = \Int$:
%
\begin{theorem} \label{thm:equiv-problems}
Let $\F\in\set{\R,\Int}$. 
The problems \ineqsolvname$(\F)$, \finineqsolvname$(\F)$ and \nonnegsolvname$(\F)$ are inter-reducible. 
All reductions are in \PTIME, except the one from \ineqsolvname$(\F)$ to \finineqsolvname$(\F)$ 
which is in \EXPTIME, but in \PTIME for fixed atom dimension.
\end{theorem}
%
%
(The proof is in Section \ref{sec:problems-proofs-one}.)
The three problems listed in Theorem~\ref{thm:equiv-problems} deserve a shared name \emph{orbit-finite linear programming} 
(in case of $\F=\R$) and \emph{orbit-finite integer linear programming} (in case of $\F=\Int$).
Figure \ref{fig:diagred} shows the reductions of Theorem~\ref{thm:equiv-problems} using dashed arrows.
%
%

As our two main results we prove that the linear programming is decidable, while the integer one is not.
Furthermore, the complexity of the decision procedure is exponential in atom dimension of the input system,
but polynomial in the number of orbits. 
This yields \EXPTIME complexity in general, and \PTIME complexity for any fixed atom dimension
of input. 
\begin{theorem} \label{thm:lp-decid}
\finineqsolvname$(\R)$
is decidable in \EXPTIME.
For fixed atom dimension, it is decidable in \PTIME.
\end{theorem} 
\begin{theorem} \label{thm:ilp-undecid}
\finineqsolvname$(\Z)$
is undecidable.
\end{theorem} 
(The proofs 
occupy Sections \ref{sec:lp-decid-proof} and
\ref{sec:ilp-undecid-proof}, respectively.)
Additionally, we settle the status of the last variant, \finnonnegsolvname$(\F)$.
The problem is decidable for $\F=\R$, as it reduces to both \finineqsolvname$(\F)$ and 
\nonnegsolvname$(\F)$ via reductions analogous to those of Theorem 
\ref{thm:equiv-problems}. We derive decidability also in case $\F = \Z$:

%
\begin{theorem} \label{thm:finnegsolv-decid}
\finnonnegsolvname$(\F)$ is decidable, for $\F\in\set{\R,\Int}$. 
\end{theorem}
\noindent
(The proof is in Section \ref{sec:problems-proofs-two}.)
In consequence of  Theorems~\ref{thm:ilp-undecid} and \ref{thm:finnegsolv-decid}, in case $\F = \Int$
the two arrows outgoing from \finnonnegsolvname$(\Int)$ in Figure \ref{fig:diagred} can
not be completed by the reverse arrows.

\begin{figure}
\begin{align*} 
\xymatrix@R=12pt@C=4pt{
*+[F]{\text{\small \finineqsolvname}(\F)} \ar@{-->}@/^/[rrr]_{\text{Thm~\ref{thm:equiv-problems}}} &&& *+[F]{\text{\small \ineqsolvname}(\F)} \ar@{-->}@/^/[lll] \ar@{-->}@/^1pc/[dd] \\
&&& \\
\text{\small \finnonnegsolvname}(\F) \ar[rrr] \ar[uu] 
&&& *+[F]{\text{\small \nonnegsolvname}(\F)} \ar@{-->}@/^1pc/[uu]_{\hspace{-0.6mm}\text{Thm~\ref{thm:equiv-problems}}} \\
&&& \\
\text{\small \finsolvname}(\F) \ar@/^/[rrr] \ar[uu] &&& \text{\small \solvname}(\F) \ar@/^/[lll]_{\text{\cite{GHL22}}}
\ar[uu] 
}
\end{align*}
\caption{Diagram of reductions between solvability problems.}
\label{fig:diagred}
\end{figure}
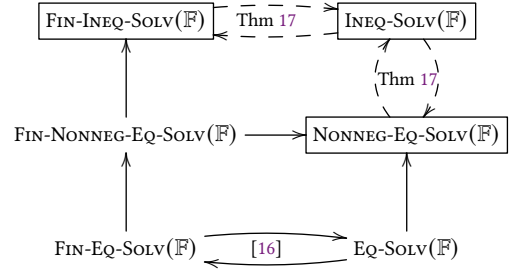

\para{Linear equations vs inequalities}

Solvability of orbit-finite systems of \emph{equations} 
(\solvname$(\F)$)
easily reduces to \ineqsolvname$(\F)$,
by replacing each equation with two opposite inequalities, 
but also to \nonnegsolvname$(\F)$, by replacing each unknown with a difference of two unknowns.
Likewise does the variant \finsolvname$(\F)$, where one only seeks for finitary solutions.
%

\begin{theorem}[\cite{GHL22} Thms 4.4 and 6.1]
\solvname$(\F)$ and \finsolvname$(\F)$ are inter-reducible and decidable%
\footnote{The results of~\cite{GHL22} apply to systems of equations where coefficients and solutions are from
any fixed commutative and effective ring $\F$. 
This includes integers $\Int$ or rationals $\Rat$ (and hence applies also to real solutions).}.
\end{theorem}

In summary, for each choice of $\F$ one may distinguish three different decision problems:
solving of systems of linear equations (two bottom nodes in Figure \ref{fig:diagred}),
solving of system of linear inequalities (three upper nodes in Figure \ref{fig:diagred}),
and the intermediate problem \finnonnegsolvname$(\F)$.

\para{Optimisation problems}

We consider $\F = \R$, due to the undecidability of Theorem~\ref{thm:ilp-undecid}.
All variants of linear programming mentioned above have corresponding \emph{maximisation} problems.
In each variant the input contains, except for a system $(\vr A, \vr t)$, 
an integer vector $\vr s : C \tofs \Z$ that represents
a (partial) linear \emph{objective} function $S : \GLin C \tofs \R$, defined by
\[
S(\vr x) = \innerprod{\vr s}{\vr x}.
\]
%
%
The maximisation problem asks to compute the supremum of the objective function over all (finitary, nonnegative)
solutions of $(\vr A, \vr t)$.
A symmetrical \emph{minimisation} problem is easily transformed to a maximisation one by replacing 
$\vr s$ with $-\vr s$.
This yields three optimisation problems 
\ineqmaxname$(\R)$, \finineqmaxname$(\R)$ and \nonnegmaxname$(\R)$ which are,
as before, inter-reducible:
%
%
\begin{theorem} \label{thm:equiv-max}
The problems \ineqmaxname$(\R)$, \finineqmaxname$(\R)$ and \nonnegmaxname$(\R)$
are inter-reducible, with the same complexity as in Theorem \ref{thm:equiv-problems}. 
\end{theorem}
(The proof is in Section \ref{sec:problems-proofs-one}.)
As our last main result we strengthen Theorem~\ref{thm:lp-decid} to the optimisation setting:
\begin{theorem} \label{thm:lp-max}
\finineqmaxname$(\R)$ is computable in \EXPTIME.
For fixed atom dimension, it is computable in \PTIME.
\end{theorem}
(The proof is in Section \ref{sec:max}.)
Hence, for every fixed atom dimension, orbit-finite linear programming is not more costly
than the classical finite linear programming.

\para{Representation of input}

There are several possible ways of representing input $(\vr A, \vr t, \vr s)$ to our algorithms.
One possibility is to rely on the equivalence between (hereditary) orbit-finite sets and
\emph{definable} sets \cite[Sect.~4]{atombook}.
We choose another standard possibility, as specified in items (1)--(3) below.
First, the representation includes:
\begin{enumerate}
\item[(1)]  a common support $T\subseteqfin\A$ of $\vr A$, $\vr t$ and $\vr s$.
\end{enumerate}
Second, knowing that $B$ and $C$ are disjoint unions of $T$-orbits, 
and relying on Lemma \ref{lem:repn}, the representation includes also:
\begin{enumerate}
\item[(2)] a list of all $T$-orbits included in $B$ and $C$, each one represented by some tuple $a\in\otu \A n$
and $G\leq S_n$; and a list of $T$-orbits included in $B\times C$, each one represented by some its element.
\end{enumerate}
Finally, relying on Lemma~\ref{lem:1O}, we assume that the representation includes also:
\begin{enumerate}
\item[(3)] a list of integer values $\orbval{\vr t}(\O)$, $\orbval{\vr s}(\O)$, and $\orbval{\vr A}(\O)$, respectively,  
for all $T$-orbits $\O$ included in $B$, $C$, and $B\times C$, respectively
(we apply Notation \ref{not:orbval}).
Integers are represented in binary.
\end{enumerate}

\para{Strict inequalities}
In this paper we consider system of \emph{non-strict} inequalities, for the sake of presentation.
The decision procedures of 
Theorems~\ref{thm:lp-decid} and \ref{thm:lp-max}, work equally well if both \emph{non-strict and strict} inequalities are
allowed.
Reductions between \finineqsolvname$(\F)$ and \ineqsolvname$(\F)$ work as well, but not the reductions
from (\finname)\nonnegsolvname$(\F)$ to (\finname)\ineqsolvname$(\F)$
as we can not simulate equalities with strict inequalities.
%


\para{Proviso}
When investigating different systems of inequalities in the following sections, we
implicitly consider their \emph{real} solutions, unless specified otherwise.


\section{Polynomially-parametrised inequalities}
\label{sec:poly}


We now introduce a core problem that will serve as a target of reductions in the proofs of Theorems \ref{thm:lp-decid} and \ref{thm:lp-max}  in Sections \ref{sec:lp-decid-proof} and \ref{sec:max}.
Consider a finite inequality $\ineqal$ of the form:
\begin{align} \label{eq:polyeq}
p_1(n) \cdot x_1 + \ldots + p_k(n) \cdot x_k \ \geq \ q(n),
\end{align}
where $p_1, \ldots, p_n$ and $q$ are univariate polynomials with integer coefficients, 
and $x_1, \ldots, x_n$ are unknowns.
The special unknown $n$ plays a role of a nonnegative integer parameter, and that is why we call such an inequality
\emph{polynomially-parametrised}.
For every fixed value $n\in\Nat$, by evaluating all polynomials in $n$ we get an ordinary
inequality  $\ineqal(n)$ with integer coefficients.
Also, if $n$ does not appear in $\ineqal$, i.e., all polynomials are constants, $\ineqal$ is an ordinary inequality.

In the sequel we study solvability of a finite system $P$ of such inequalities~\eqref{eq:polyeq} 
with the same unknowns $x_1, \ldots, x_k$.
Again,
by evaluating all polynomials in $n$ we get an ordinary system $P(n)$.
We use the matrix form $P(n) = (\vr A(n), \vr t(n))$ when convenient.
A fundamental problem is to check if for some value $n\in\Nat$, the system $P(n)$ has a real solution:

\prob{\polyineqsolvname}
{A finite system of polynomially-parametrised inequalities $P$} 
{Does $P(n)$ have a real solution for some $n\in \Nat$}
%

\begin{theorem} \label{thm:polydecid}
\polyineqsolvname is decidable. 
\end{theorem}
\noindent
(The proof is in Section \ref{sec:poly-proofs}.)
In the sequel we will not use the decision procedure of Theorem \ref{thm:polydecid}, but rather the algorithm of Theorem \ref{thm:polyptime} stated below, since our
later applications only use \emph{monotonic} instances of \polyineqsolvname.

\subsection{Monotonic polynomially-parametrised inequalities}
\label{sec:poly-mon}

A system $P$ is \emph{monotonic} if there is some $n_0\in\Nat$ such that
every solution of $P(n)$, for an integer $n\geq n_0$, 
is also a solution of $P(n+1)$.
Note that monotonicity vacuously holds (with any value of $n_0$) 
if $n$ does not appear in $P$, i.e., when $P$ is an ordinary (non-parametrised) system.
When $P$ is monotonic, 
a solution of $P(n)$ for $n\geq n_0$ is also a solution of $P(n')$ for all integers $n'\geq n$.
%
%
Therefore, instead of \polyineqsolvname we prefer to use in the sequel
the following core problem, where we do not assume monotonicity but
seek for a solution of $P(n)$ for \emph{almost all} 
(all sufficiently large) values of the parameter $n\in\Nat$:

\prob{\allpolyineqsolvname}
{A finite system of polynomially-parametrised inequalities $P$} 
{Is there $n_0 \in \Nat$ and a vector $(x_1, \ldots, x_k)$ which is a solution of $P(n)$ for every integer $n\geq n_0$}
%
%
\noindent
From now on, a vector $(x_1, \ldots, x_k)$ which is a solution of an inequality $\ineqal(n)$ 
(resp.~a system $P(n)$) for almost all $n\in\Nat$ we call \emph{\aasol} of 
$\ineqal$ (resp.~$P$).
The rest of this section is devoted to designing a \PTIME algorithm for \allpolyineqsolvname:

\begin{theorem} \label{thm:polyptime}
\allpolyineqsolvname is 
in \PTIME.
\end{theorem}
\begin{proof}
Consider a polynomially-parametrised 
inequality $\ineqal$ of the form:
\begin{align} \label{eq:polyeqpowt}
p_1(n) \cdot x_1 + \ldots + p_k(n) \cdot x_k \ \geq \ q(n).
\end{align}
Let $d$  be the maximal degree of polynomials $p_1, \ldots, p_k, q$ appearing
in $\ineqal$.
We call $d$ \emph{the degree} of $\ineqal$, and denote it also as $\deg \ineqal$.
Let $a_1, \ldots, a_k, b$ be (integer) coefficients of the monomial $n^{d}$ in 
$p_1,  \ldots, p_k, q$,
respectively.
Therefore
\begin{align} \label{eq:d}
\begin{aligned}
p_1(n) & = a_1 \cdot n^{d} + p'_1(n) \qquad \ldots \qquad 
p_k(n) = a_k \cdot n^{d} + p'_k(n) \qquad\qquad
q(n) & = b \cdot n^{d} + q'(n)
\end{aligned}
\end{align}
for some polynomials $p'_1, \ldots, p'_k, q$ of degree strictly smaller than $d$.
The ordinary inequality with integer coefficients
\begin{align} \label{eq:head}
a_1 \cdot x_1 + \ldots + a_k \cdot x_k \geq b,
\end{align}
we call \emph{the head inequality} of $\ineqal$, and denote by
$\li{\ineqal}$.
Furthermore, the polynomially-parametrised inequality 
\begin{align} \label{eq:tl}
p'_1(n) \cdot x_1 + \ldots + p'_k(n) \cdot x_k \ \geq \ q'(n),
\end{align}
obtained by removing all appearances of the monomial $n^d$,
we call the \emph{tail} of $\ineqal$, and denote it by $\tl \ineqal$.
We also consider below the strict strengthening of the head inequality \eqref{eq:head}, denoted as $\strictli \ineqal$, and the equality, denoted as $\lipar \ineqal =$:
\begin{align} \label{eq:strictli}
%
a_1 \cdot x_1 + \ldots + a_k \cdot x_k > b,
\qquad\qquad\qquad
a_1 \cdot x_1 + \ldots + a_k \cdot x_k = b.
\end{align}
%
%
As $\ineqal$ is equal to the sum of its head $\li \ineqal$ multiplied by $n^d$, 
and its tail $\tl \ineqal$, we immediately deduce:
\begin{claim} \label{claim:hdtl}
For every $n\in\Nat$, every solution of $\lipar {\ineqal} =$ is either a solution of both $\ineqal(n)$ and 
${\tl \ineqal}(n)$, 
or of none of them.
\end{claim}

We now provide under- and over-approximations
of the solution set of $\ineqal$ (in Claims \ref{claim:strictlibase} and \ref{claim:libase}).

\begin{claim} \label{claim:libase}
Every \aasol of $\ineqal$ is also a solution of $\li \ineqal$.
\end{claim}
\begin{proof}
%
Consider an inequality $\ineqal$ \eqref{eq:polyeqpowt}
and  its \aasol  $\vr x = (x_1, \ldots, x_k)$.
Let $d = \deg \ineqal$.
We thus have
\[
\frac{p_1(n)}{n^d} \cdot x_1 + \ldots + \frac{p_k(n)}{n^d} \cdot x_k \ \geq \ \frac{q(n)}{n^d}
\]
for all sufficiently large $n\in\Nat$. 
Using the decomposition \eqref{eq:d}, we rewrite the above inequality to
\[
\Big(a_1 + \frac{p'_1(n)}{n^d}\Big) \cdot x_1 + \ldots + 
\Big(a_k + \frac{p'_k(n)}{n^d} \Big) \cdot x_k \ \geq \ 
b + \frac{q'(n)}{n^d}.
\]
As the degrees of all polynomials $p'_1, \ldots, p'_k, q'$ are smaller than $d$, 
all the fractions tend to 0 when $n$ tends to $\infty$, and we may deduce
\[
a_1 \cdot x_1 + \ldots + a_k \cdot x_k \ \geq \  b,
\]
i.e., $\vr x$ is a solution of $\li \ineqal$, as required.
%
\end{proof}

\begin{claim} \label{claim:strictlibase}
Every solution of $\strictli {\ineqal}$ is also an \aasol of $\ineqal$.
\end{claim}
\begin{proof}
Let $d = \deg \ineqal$. 
Consider any vector $\vr x = (x_1, \ldots, x_k)$  satisfying
the strict inequality $\strictli {\ineqal}$ (in \eqref{eq:strictli} on the left). 
Therefore for any polynomials $p'_1, \ldots, p'_k, q'$ of degree strictly smaller than $d$,
the inequality
\[
\Big(a_1 + \frac{p'_1(n)}{n^d}\Big) \cdot x_1 + \ldots + 
\Big(a_k + \frac{p'_k(n)}{n^d} \Big) \cdot x_k > 
b + \frac{q'(n)}{n^d}
\]
is satisfied for all sufficiently large $n\in\Nat$.
Applying the above inequality to polynomials appearing in \eqref{eq:d}, we obtain:
\[
\frac{p_1(n)}{n^d} \cdot x_1 + \ldots + \frac{p_k(n)}{n^d} \cdot x_k \ > \ \frac{q(n)}{n^d}
\]
for all sufficiently large $n\in\Nat$.
We multiply both sides by $n^d$ in order to derive that 
$\vr x$ is a solution of $\ineqal(n)$ for all sufficiently large $n\in\Nat$,
as required.
\end{proof}


Consider an instance $P$ of \polyineqsolvname, i.e., a finite system of 
polynomially-parametrised inequalities of the form \eqref{eq:polyeqpowt}.
Let $\li{P} := \setof{\li{\ineqal}}{\ineqal\in P}$ be the system of head inequalities
(note that degrees of different inequalities in $P$ may differ), and let
$\strictli{P} := \setof{\strictli{\ineqal}}{\ineqal\in P}$.
Using Claims \ref{claim:strictlibase} and \ref{claim:libase} we derive:

\begin{claim} \label{claim:li}
Every \aasol of $P$ is also a solution of \xspace $\li P$.
\end{claim}

\begin{claim} \label{claim:strictli}
Every solution of $\strictli P$ is also an \aasol of $P$.
\end{claim}

For time estimation, as the size measure $\size \ineqal$ of an inequality $\ineqal$ we take the total number of monomials
appearing in $\ineqal$.
In particular, $\size \ineqal > \size {\tl \ineqal}$.
The size of a system $P$ is the sum of sizes of all its inequalities.
For two systems $P'$, $P''$ of inequalities, we denote their union by $P'\conj P''$
(clearly, union of systems corresponds to conjunction of constraints).
We write $P\conj\ineqal$ instead of $P\conj\set{\ineqal}$.
By $P\setminus \ineqal$ we denote the system obtained from $P$ by removing an inequality $\ineqal$.

\para{The algorithm}
A decision procedure for \allpolyineqsolvname iteratively transforms an instance
of the form $P \conj \G$, where $P$ is a system
of polynomially-parametrised inequalities, and $\G$ is a system of ordinary (non-parametrised)
equalities over the same unknowns.
Initially, $\G$ is empty.
We define a transformation step that given such an instance $P \conj\G$, 
either confirms its solvability (existence of an \aasol), 
or confirms its non-solvability (non-existence of an \aasol), or outputs an instance $P' \conj \G'$ 
which has the same {\aasol}s as $P\conj \G$, and 
such that $\size {P'} < \size P$.
\allpolyineqsolvname is solved by iterating the transformation step until it confirms
either solvability or non-solvability.
Termination after a polynomial number of iterations is guaranteed, as
$\size P$, while being nonnegative, strictly decreases in each iteration.
The transformation step invokes a \PTIME procedure 
for 
ordinary linear programming
(as detailed in \eqref{eq:>=} and \eqref{eq:=} below).
Here is a pseudo-code of the algorithm:

\begin{algorithm}[H]
\caption{(\allpolyineqsolvname)}
\label{alg:1}
\begin{algorithmic}[1]
\State \textbf{Input:} \ A polynomially-parametrised system $P$.
\smallskip
\State $\G \gets \emptyset$
\Repeat
\smallskip
\If{\ $\li P  \conj  \G$ \ \eqref{eq:>=} \  is non-solvable} \label{c:if1}
\Comment see Claim \ref{claim:li}
\State report non-solvability of $P$
\Else
\If{$\strictli \ineqal \conj \li {P\setminus \ineqal} \conj \G$ \ \eqref{eq:=} \ is solvable for all $\ineqal \in P$}
\label{c:if2}
\Comment see Claims \ref{claim:strictli}, \ref{claim:=}
\State report solvability of $P$
\Else
\State \textbf{choose any} $\ineqal \in P$ \textbf{such that} \ $\strictli \ineqal \conj \li {P\setminus \ineqal} \conj \G$ \ \eqref{eq:=} \ is non-solvable \label{l:lastif}
\Comment see Claim \ref{claim:samesol}
\smallskip
\State $P \, \gets \, (P \setminus \ineqal) \ \conj \ \tl \ineqal$
\State $\G \,\, \gets \, \G \ \conj \ \lipar \ineqal {=}$
\EndIf 
\EndIf
\Until{solvability or non-solvability of $P$ is reported}
\end{algorithmic}
\end{algorithm}                                                                                                                                                                  


\para{Transformation step}
The step, defined by the body of the 
\textbf{repeat} loop, proceeds as follows.
If the ordinary system
\begin{align} \label{eq:>=}
\li P \; \conj \; \G
\end{align}
is non-solvable, 
non-solvability of $P\conj\G$ is reported.
This is correct due to Claim \ref{claim:li}.
Otherwise, knowing that \eqref{eq:>=} is solvable,
the algorithm checks, for every $\ineqal\in P$, whether 
the strengthened system 
\begin{align} \label{eq:=}
\strictli \ineqal \; \conj \; \li {P\setminus \ineqal} \; \conj \; \G,
\end{align}
obtained from \eqref{eq:>=} by replacing
the inequality $\li \ineqal$ with $\strictli \ineqal$, is also solvable.
If this is the case, solvability of $P \conj \G$ is reported. 
This is correct due to Claim \ref{claim:strictli} combined with the following one:
\begin{claim} \label{claim:=}
Solvability of \eqref{eq:=} for every inequality $\ineqal$ in $P$, implies solvability of
\begin{align}
\label{eq:>}
\strictli P \, \conj \, \G.
\end{align}
\end{claim}
\begin{proof}
Let $m$ be the number of inequalities in $P$, and
suppose that for every inequality $\ineqal$ in $P$, the system \eqref{eq:=}
has a solution, $\vr x_\ineqal$.
All $\vr x_\ineqal$ are thus solutions of \eqref{eq:>=}, and since
the solution set of \eqref{eq:>=} is convex,
the average of all these solutions
$
\frac 1 {m} \cdot \sum_{\ineqal\in P} \vr x_\ineqal
$
is then a solution of $\strictli P \, \conj \, \G$.
\end{proof}
Otherwise, we know that some
inequality $\ineqal$ in $P$ is \emph{degenerate}, namely
\eqref{eq:=} is non-solvable.
In other words, the equality $\lipar \ineqal =$
is implied by \eqref{eq:>=}.
The algorithm chooses a degenerate inequality $\ineqal\in P$ 
%
%
and creates a new instance $P' \conj \G'$, where
\begin{align*}
P' \  = \ (P \setminus \ineqal) \; \conj\; \tl \ineqal \qquad\qquad
\G' \  = \ \G \; \conj \; \lipar \ineqal {=}.
\end{align*}
In words, $P'$ is obtained from $P$ by replacing $\ineqal$ with $\tl \ineqal$,
and $\G'$ is obtained from $\G$ by adding $\lipar \ineqal =$.
As $\size {\tl \ineqal} < \size \ineqal$, we have $\size {P'} < \size P$, as required.
This completes description of the transformation step.

\para{Correctness}
By Claim \ref{claim:hdtl} we derive:
\begin{claim}  \label{claim:samesol}
Systems $P\conj\G$ and $P'\conj\G'$ have the same {\aasol}s.
\end{claim}
\begin{proof}
In one direction, 
consider an \aasol $\vr x$ of $P'\conj\G'$.
It is trivially a solution of $\G$. 
Furthermore, being a solution of $\lipar \ineqal =$ and of $\tl \ineqal(n)$ for almost all $n\in\Nat$,
by Claim \ref{claim:hdtl} it is a solution of $\ineqal(n)$ for almost all $n$, and hence an \aasol of $P$.

Conversely, 
consider an \aasol $\vr x$ of $P\conj\G$.
By Claim \ref{claim:li}, it is a solution of $\li P\conj\G$ and hence, 
as $\ineqal$ is degenerate,
also a solution of $\lipar \ineqal =$.
Therefore $\vr x$ is a solution of $\G'$.
Furthermore, being a solution of $\lipar \ineqal =$ and of $\ineqal(n)$ for all sufficiently large
$n\in\Nat$, 
by Claim \ref{claim:hdtl} it is also a solution of $\tl \ineqal(n)$ for all sufficiently large $n\in\Nat$, 
and hence an \aasol of $P'$.
\end{proof}

\para{Complexity}
We note that the main loop of the algorithm always terminates, at latest when $P=\emptyset$, as in this case the system
\eqref{eq:=} is vacuously solvable for all $\ineqal\in P$.
%
%
Solvability of \eqref{eq:>=} in line \ref{c:if1} 
is checked by one solvability test of an ordinary system of inequalities.
Solvability of \eqref{eq:=} in line \ref{c:if2} is also checkable in polynomial time due to the following claim
applied to $Q = \li {P} \; \conj \; \G$:
\begin{claim} 
Given an ordinary system $Q$ of linear inequalities and $\ineqal \in Q$, one can check,
in \PTIME,
solvability of \xspace $\ineqal_>  \conj \, (Q\setminus\ineqal)$, where $\ineqal_>$ is the strict
strengthening of $\ineqal$.
\end{claim}
\begin{proof}
We invoke ordinary linear programming twice 
(in \PTIME, see e.g.~\cite[Section 8.7]{PSbook82}).
Let $\ineqal$ be of the form $a_1 \cdot x_1 + \ldots + a_k \cdot x_k \geq b$.
If $Q$ is non-solvable, the algorithm reports non-solvability of 
$\ineqal_>  \conj \, (Q\setminus\ineqal)$.
Otherwise, the algorithm computes the supremum $M\in\Rat\cup \set{\infty}$ of the objective function 
\[
S(x_1, \ldots, x_k) = a_1 \cdot x_1 + \ldots + a_k \cdot x_k,
\]
constraint by $Q\setminus\ineqal$, by invoking ordinary linear programming.
By solvability of $Q$ we know that $M\geq b$.
If $M > b$, the algorithm reports solvability, otherwise it reports non-solvability.
\end{proof}

Number of iterations of transformation step is polynomial (as $\size P$ decreases in each iteration) and hence
so is the number of inequalities in $\G$.
In consequence, the number of invocations of ordinary linear programming is polynomial in
each transformation step, and hence polynomial in total, and each its instance
of ordinary linear programming is also polynomial.
Summing up, our decision procedure for \allpolyineqsolvname works in \PTIME.

\smallskip
The proof of Theorem \ref{thm:polyptime} is thus completed.
%
\end{proof}

\begin{slremark}
We do not need any explicit bound on the threshold value of $n_0$ guaranteeing that 
every \aasol of $P$ is a solution of $P(n)$ for every integer $n \geq n_0$.
On the other hand, an exponential bound is derivable from our algorithm.
Assuming $P$ has an \aasol, $P$ has also an \aasol $\vr x$ which is at most exponentially large, 
e.g., a solution of an ordinary system \eqref{eq:>}.
Substituting $\vr x$ into $P(n)$ yields a system of univariate polynomial inequalities, and 
one can take as threshold $n_0$ any integer larger than all nonnegative roots of all polynomials appearing
in the system. 
As roots of univariate polynomials are polynomially bounded, we deduce the bound for $n_0$.
\end{slremark}

\begin{slexample}
Recall two polynomially-parametrised inequalities \eqref{eq:ppar} in Example \ref{ex:poly}
in Section \ref{sec:intro}. They have the same head inequality
$
x\geq 1,
$
which is trivially solvable, and hence the algorithm reports solvability after the first iteration.
Both the tail inequalities,
$-x \geq 1$ and
$0 \geq 1$,
are ordinary (non-parametrised).

The following instance $P_0$ admits three iterations of the main loop of the algorithm:
\begin{align*}
n^2 \cdot x \ - \ n^2 \cdot y \  +  \ n \cdot z  & \ \geq \ 0 \\
-n \cdot x \ + \ (n+3) \cdot y \ \ \ \ \ \, & \ \geq \ 0
\end{align*}
The head inequalities of these two inequalities are $x-y \geq 0$ and $-x + y \geq 0$, respectively.
Therefore
the system $\li {P_0}$ is equivalent to $x=y$ and hence solvable, while $\strictli {P_0}$ is not,
and both inequalities in $P_0$ are degenerate.
Supposing the first one is chosen by the algorithm, after the first iteration we get the following systems $P_1$
(left) and $\G_1$ (right): 

\begin{minipage}{0.5\linewidth}
\begin{align*}
n \cdot z  & \ \geq \ 0\\
-n \cdot x \ + \ (n+3) \cdot y  & \ \geq \ 0
\end{align*}
\vspace{-2mm}
\end{minipage}
\begin{minipage}{0.5\linewidth}
\begin{align*}
x \ - \ y  & \ = \ 0 \\
\end{align*}
\vspace{-2mm}
\end{minipage}

\noindent
In the second iteration, 
the system $\li{P_1} \conj \G_1$ (left) is solvable but the system $\strictli{P_1} \conj \G_1$ (right) 
is not:

\begin{minipage}{0.5\linewidth}
\begin{align*}
z  & \ \geq \ 0\\
- x \ + \  y  & \ \geq \ 0 \\
x \ - \ y  & \ = \ 0 
\end{align*}
\vspace{-2mm}
\end{minipage}
\begin{minipage}{0.5\linewidth}
\begin{align*}
z  & \ > \ 0\\
- x \ + \  y  & \ > \ 0 \\
x \ - \ y  & \ = \ 0  
\end{align*}
\vspace{-2mm}
\end{minipage}

\noindent
The algorithm picks up the second inequality in $P_1$, the only degenerate one, and sets
$P_2$ (left) and $\G_2$ (right):

\begin{minipage}{0.5\linewidth}
\begin{align*}
n \cdot z  & \ \geq \ 0\\
3 \cdot y  & \ \geq \ 0
\end{align*}
\vspace{-2mm}
\end{minipage}
\begin{minipage}{0.5\linewidth}
\begin{align*}
x \ - \ y  & \ = \ 0 \\
-x  \ + \ y & \ = \ 0
\end{align*}
\vspace{-2mm}
\end{minipage}

\noindent
In the last third iteration, the system $\strictli {P_2} \conj \G_2$ 
(obtained by replacing the inequality $n\cdot z\geq 0$ by $z > 0$)
is solvable, and hence solvability of $P_0$ is reported.
\end{slexample}

%
%

%

%
%


\section{Finitely setwise-supported sets} \label{sec:weak}

In this section
we introduce the novel concept of \emph{setwise-support}, playing a central role in the proofs of Theorems \ref{thm:lp-decid} and \ref{thm:lp-max}.
In short, we replace \emph{pointwise} stabilisers by \emph{setwise} ones. 
%

For any $T\subseteqfin\A$ consider the set of all atom automorphisms
that preserve $T$ as a set only (called \emph{setwise-$T$-automorphisms}):
\[
\Aut {\set{T}} \ = \ \setof{\pi \in \Aut {}}{\pi(T) = T}.%
\footnote{
$\Aut {\set{T}}$ is often called the \emph{setwise stabilizer} of $T$.
}
\]
Accordingly, we define \emph{setwise-$T$-orbits} as equivalence classes with respect to the action of $\Aut {\set{T}}$:
two 
sets (elements) $x, y$ are in the same setwise-$T$-orbit if $\pi(x) = y$ for some $\pi \in\Aut {\set{T}}$.
We have
\[
\Aut {T} \subseteq \Aut{\set{T}} \subseteq \Aut {},
\]
and hence every equivariant orbit splits into finitely many setwise-$T$-orbits, 
each of which splits in turn into finitely many $T$-orbits.
A set $X$ is \emph{setwise-$T$-supported} if $\pi(X) = X$ for all $\pi \in \Aut {\set{T}}$.
Equivalently, $X$ is a union of setwise-$T$-orbits.
Note that each setwise-$T$-supported set is $T$-supported, but the opposite implication is not true.
When $T$ is irrelevant, we speak of finitely setwise-supported sets.
Finally notice that a setwise-$T$-supported set is not necessarily setwise-$T'$-supported for $T\subseteq T'$,
which distinguishes setwise-support from  standard support.

\begin{slexample}
\label{ex:av}
Let $T=\set{\a,\b}\subseteq\A$.
The vector $\vr v$, defined in Example \ref{ex:vect} in Section \ref{sec:problems}, is not setwise-$T$-supported.
Indeed,
\[
\pi(\vr v)(\a, \x) = \vr v(\b, \x) \neq \vr v(\a, \x) 
\]
for any $\x \notin T$ and $\pi \in \Aut{\set{T}}$ that swaps $\a$ and $\b$ but preserves all other atoms.
The \emph{averaged} vector $\vr v'$ defined by
\begin{align*}
\vr v'(\a \x) &= \vr v'(\x \a) = - 1.5 &
\vr v'(\a \b) &= \vr v'(\b \a) = 3  \\ 
\vr v'(\b \x) &= \vr v'(\x \b) = - 1.5
&
\vr v'(\x \g) &= 0,
\end{align*}
for  $\x,\g\in\A\setminus\set{\a,\b}$, is setwise-$T$-supported.
Notice that $\vr v'$, is not setwise-$(T\cup\set{\g})$-supported, for $\g\notin T$.
\end{slexample}

Clearly, with the size of $T$ increasing towards infinity, the number of $T$-orbits 
included in one equivariant orbit may increase towards infinity as well.
The crucial property of setwise-$T$-supported sets is that they do not suffer from this unbounded growth:
the number of setwise-$T$-orbits included in a fixed equivariant orbit is bounded, no matter how large $T$ is.
We will need this property for setwise-$T$-orbits $U\subseteq\otu {\A} n$,  $n\in\Nat$, and
it follows immediately by Lemma~\ref{lem:STorbit}.
Intuitively speaking, each such setwise-$T$-orbit is determined by 
a subset $I\subseteq \setto n$ of positions which is filled by arbitrary pairwise different atoms from $T$, 
the remaining positions $\setto n\setminus I$ are filled by arbitrary atoms from $\A\setminus T$ 
(cf.~Lemma \ref{lem:tupleorbit} in Section \ref{sec:atoms}).

\begin{lemma} \label{lem:STorbit}
Let $T\subseteqfin\A$ of size $\size T\geq n$.
Each setwise-$T$-orbit $\O \subseteq \otu {\A} n$ is of the form
\begin{align} \label{eq:STorbit}
\begin{aligned}
\O =  \setof{a \in \otu \A n}{\proj n I(a) \in \otu T \ell, \ \  
  \proj n {\setto n \setminus I}(a) \in \otu {\big(\A\setminus T\big)} {n-\ell}},
\end{aligned}
\end{align}
for some $I\subseteq \setto n$ of size $\ell$.
\end{lemma}
\begin{proof}
Consider any tuple $t = (\a_1, \ldots, \a_n) \in \otu {\A} n$.
Let $I = \setof{i \in \setto n}{\a_i \in T}$ denote the positions in $t$ filled by atoms from $T$.
By applying all setwise-$T$-automorphisms to $t$, we obtain all tuples, where positions from $I$ are arbitrarily 
filled by elements of $T$, 
and positions outside of $I$ are arbitrarily
filled by elements of $\A\setminus T$.
\end{proof}
%

%
%
%

\begin{notation} \label{not:orbsum}
Given a finitary vector $\vr x : C\tofin\R$ and an equivariant orbit $\O\subseteq C$, 
we write
\[
\orbsum {\vr x}(\O) \ = \ \sum_{c\in \O} \vr x(c)
\]
to denote for the sum of $\vr x(c)$ ranging over all $c\in \O$.
This yields the finite \emph{orbit-sum} vector
\[
\orbsum {\vr x} : \orbits C \to \R
\]
mapping the equivariant orbits included in $C$ to $\R$.
\end{notation}

A key observation is that a solvable equivariant system necessarily has a finitely setwise-supported solution:
\begin{lemma} \label{lem:weaksol}
If an equivariant system of inequalities $(\vr A, \vr t)$ has a finitary $T$-supported solution 
$\vr x$ then it also has a finitary
setwise-$T$-supported one $\vr y$ such that
$\orbsum{\vr x} = \orbsum{\vr y}$.
\end{lemma}
%
%
\begin{proof}
Let $\vr x : C\tofs \R$ be a solution of the system, namely $\innerprod{\vr A}{\vr x}\geq \vr t.$
Let $T = \supp {\vr x}$ and $n= \size T$.
As $(\vr A, \vr t)$ is equivariant,
atom automorphisms 
preserve being a solution, namely
for every $\rho \in \Aut {}$, the vector $\rho(\vr x)$ is also a solution:
$
\innerprod {\vr A} \rho(\vr x) \ \geq \  \vr t.
$
Consider $\Aut{\A\setminus T}$, the subgroup of atom automorphisms that only permute $T$ and preserve
all other atoms.
Knowing that the size of $\Aut{\A\setminus T}$ is $n!$, we have
\[
\innerprod {\vr A} \Big(\sum_{\rho \in \Aut{\A\setminus T}}\rho(\vr x)\Big) \ \geq \ n! \cdot \vr t,
\]
and hence the vector $\vr y$ defined by averaging (cf.~Example \ref{ex:av})
\begin{align} \label{eq:defy}
\vr y \ = \ \frac{1}{n!} \ \cdot \ \sum_{\rho \in \Aut{\A\setminus T}} \rho(\vr x)
\end{align}
is also a solution of the system, namely $\innerprod {\vr A} \vr y  \geq  \vr t$.
We notice that for finitary $\vr x$, the vector $\vr y$ is finitary as well.
By the very definition, the averaging \eqref{eq:defy} preserves the orbit-sum:
$\orbsum{\vr x} = \orbsum{\vr y}$.
Furthermore, we claim that the vector $\vr y$ is setwise-$T$-supported.
To prove this, we fix an arbitrary $\pi \in \Aut {\set{T}}$, aiming at showing
that $\pi(\vr y) = \vr y$.
It factors through
$
\pi = \sigma \circ \rho
$
for some $\rho \in \Aut{\A\setminus T}$ and $\sigma \in \Aut T$.
Indeed, $\rho$ acts as $\pi$ on $T$ but is identity elsewhere, while
$\sigma$ acts as $\pi$ outside of $T$ but is identity on $T$.
%
%
A crucial but simple observation is that, 
by the very construction of $\vr y$, we have 
\begin{align} \label{eq:1}
\rho(\vr y) = \vr y.
\end{align}
Indeed, as $\vr y$ is defined by averaging over all $\rho'\in\Aut{\A\setminus T}$, 
\[
\rho \Big(\sum_{\rho' \in \Aut{\A\setminus T}}\rho'(\vr x)\Big) \ =
\sum_{\rho' \in \Aut{\A\setminus T}}\rho\circ\rho'(\vr x) \ =
\sum_{\rho' \in \Aut{\A\setminus T}}\rho'(\vr x)
\]
which implies $\rho(\vr y) = \vr y$.
Moreover, as action of atom automorphisms commutes with support, we have
\[
\supp {\rho'(\vr x)} = \rho'(\supp{\vr x})
\]
for every $\rho' \in \Aut{}$, and therefore 
\[
\supp {\rho'(\vr x)} = \supp{\vr x}
\]
for every $\rho' \in \Aut{\A\setminus T}$. 
Therefore $T$ supports the right-hand side of \eqref{eq:defy}, which means that
$\supp {\vr y} \subseteq T$ and implies
%
%
\begin{align} \label{eq:2}
\sigma(\vr y) = \vr y.
\end{align}
By \eqref{eq:1} and \eqref{eq:2} we obtain $\pi(\vr y) = \vr y$, as required.


Finally, 
the equality $\orbsum {\vr x} = \orbsum {\vr y}$ follows directly by \eqref{eq:defy}.
%
\end{proof}

\begin{slexample}
\label{ex:Kirchoff-cont-cont}
Recall the system of inequalities from Examples \ref{ex:Kirchoff} and \ref{ex:Kirchoff-cont}.
Its finitary solutions correspond to finite directed graphs, whose vertices and edges are labeled
by real numbers satisfying constraints \eqref{eq:K1} and \eqref{eq:K2}.
According to Lemma \ref{lem:weaksol}, if such a directed graph existed,  
there would also exist a directed clique, where labels of all vertices are pairwise equal, 
and labels of all edges are pairwise equal as well, which satisfying constraints \eqref{eq:K1} and \eqref{eq:K2}.
In particular, all edges incoming to a vertex would carry the same value as all outgoing edges.
This requirement is clearly contradictory with constraints \eqref{eq:K1} and \eqref{eq:K2},
and hence the system has no finitary solutions.
\end{slexample}

%
%

In the next section we rely on the fact that
existence of a setwise-$S$-supported solution implies existence 
of such a solution for any support larger than $S$.
The fact follows immediately from
Lemma \ref{lem:weaksol}, since every setwise-$S$-supported vector is trivially $T$-supported, for
every superset $T$ of $S$:
\begin{corollary} \label{cor:weaksol}
If an equivariant system of inequalities $(\vr A, \vr t)$ has a finitary setwise-$S$-supported solution $\vr x$,
then  for every superset $T$ of $S$ of size $\size{T} = \size S +1$, 
the system $(\vr A, \vr t)$ has a finitary setwise-$T$-supported solution $\vr y$ 
such that $\orbsum {\vr x} = \orbsum {\vr y}$.
\end{corollary}
%


\section{Decidability of real solvability}
\label{sec:lp-decid-proof}

In this section
we prove Theorem~\ref{thm:lp-decid} by a reduction of \finineqsolvname$(\R)$ to
\polyineqsolvname
(cf.~Example \ref{ex:poly} in Section \ref{sec:intro}).

\subsection{Preliminaries}
\label{sec:prelimproof}

Consider an orbit-finite system of inequalities given by a matrix $\vr A :  B\times C \tofs \Z$ and
a target vector $\vr t : B\tofs \Z$. 
%
%

\begin{lemma} \label{lem:wlog}
W.l.o.g.~we can assume that $B$ and $C$ are disjoint unions of equivariant orbits $\otu \A k$, $k\in\Nat$:
\begin{align} \label{eq:eqorbits}
\begin{aligned}
B =  
\otu {\A} {n_1}  \uplus  \ldots \uplus  \otu {\A} {n_s} \qquad\qquad 
C =  \otu {\A} {m_1} \uplus \ldots \uplus \otu {\A} {m_r}
\end{aligned}
\end{align}
(see the figure below),
and that $\vr A$ and $\vr t$ are equivariant. The size blow-up is exponential in
atom dimension, but polynomial when atom dimension is fixed.
\end{lemma}
%
\vspace{-4mm}
\begin{align*}
& \qquad\  \otu \A {m_1} \, \otu \A {m_2} \cdots \, \otu \A {m_r} \\
\vr A \ = \ \ \ &
\begin{matrix}
\otu \A {n_1} \\ \otu \A {n_2} \\ \cdots \\ \otu \A {n_s}
\end{matrix}
\left[
\begin{matrix}
\ \ \ \ \ \ \ \\ \hline
\\ \hline
\\ \hline
\ \ 
\end{matrix}
\right\rvert
\hspace{-1.25mm}
\left\rvert
\begin{matrix}
\ \ \ \ \ \ \  \\ \hline
\\ \hline
\\ \hline
\ \ 
\end{matrix}
\right\rvert
\hspace{-1.25mm}
\left\rvert
\begin{matrix}
\ \ \ \ \ \ \ \\ \hline
\\ \hline
\\ \hline
\ \ 
\end{matrix}
\right\rvert
\hspace{-1.25mm}
\left\rvert
\begin{matrix}
\ \ \ \ \ \ \ \\ \hline
\\ \hline
\\ \hline
\ \ 
\end{matrix}
\right]
\qquad
\vr t \ = \ \begin{bmatrix}
\ \ \\ \hline
\\ \hline
\\ \hline
\ \ 
\end{bmatrix}
\end{align*}
\vspace{1mm}

\noindent
(The proof is in Section \ref{sec:lp-decid-proof-proofs}.)
Note that this includes the case of finite systems, namely
$n_1 = \ldots = n_s = m_1 = \ldots = m_r = 0$.

\subsection{Idea of the reduction}
%
Suppose only finitary \emph{$T$-supported} solutions are sought, for a fixed $T\subseteqfin\A$.
\finineqsolvname$(\R)$ reduces then to a finite system of inequalities $(\vr A', \vr t')$ obtained
from $(\vr A, \vr t)$ as follows:
\begin{enumerate}
\item[(1)] Keep only columns indexed by $T$-tuples (= elements of \emph{finite} $T$-orbits) $c\in C$, discarding all other columns.
\item[(2)]
Pick arbitrary representatives of \emph{all} $T$-orbits included in $B$, and  keep only rows of $\vr A$
and entries of $\vr t$ indexed by the representatives, discarding all others.
\end{enumerate}
The system $(\vr A', \vr t')$ is solvable if and only if the original one $(\vr A, \vr t)$ has
a finitary $T$-supported solution.
Indeed, discarding unknowns as in (1) is justified as a finitary $T$-supported solution of $(\vr A, \vr t)$ 
assigns $0$ to each non-$T$-tuple.
Discarding inequalities as in (2) is also justified.
Indeed, each inequality in the original system is obtained by applying some atom $T$-automorphism to an inequality
in $(\vr A', \vr t')$, while atom $T$-automorphisms preserve $T$-supported solutions of $(\vr A', \vr t')$,
which implies that every $T$-supported solutions of $(\vr A', \vr t')$ 
is also a solution of all inequalities in the original system.
%
%
%

The above reduction yields no algorithm yet, as
we do not know a priori any bound on size of $T$,
and
the size of $(\vr A', \vr t')$ depends on the number of $T$-orbits and hence grows unboundedly when  $T$ grows.
We overcome this difficulty by using setwise-$T$-orbits instead of $T$-orbits, and relying on Lemmas \ref{lem:weaksol}
and \ref{lem:STorbit}.
The latter one guarantees that the number of setwise-$T$-orbits is constant - independent of $T$.
Once we additionally merge (sum up) all columns indexed by elements of the same setwise-$T$-orbit, we get $\vr A'$ of size independent of $T$.


This still does not yield an algorithm, as entries of $\vr A'$ change
when  $T$ grows.
We however crucially discover 
that the growth of the entries of $\vr A'$ is \emph{polynomial}
in $n = \size T$, for sufficiently large $n$.
Therefore, $\vr A'$ is a matrix of polynomials in one unknown $n$, and solvability of $(\vr A, \vr t)$ is equivalent to solvability
of $(\vr A', \vr t')$ for some value $n\in \Nat$.
As argued in Section \ref{sec:poly}, the latter solvability is decidable.


\subsection{Reduction of \finineqsolvname$(\R)$ to \allpolyineqsolvname}
\label{sec:red2poly}

Let us fix an equivariant system $(\vr A, \vr t)$.
We construct a finite system $P_2$ of polynomially-parametrised 
inequalities such that $(\vr A, \vr t)$ has a finitary solution
 if and only if $P_2(n)$ has a solution for almost all $n\in\Nat$.

Let us denote by $d = \max\set{n_1, \ldots, n_s, m_1, \ldots, m_r}$ the maximal atom dimension of
orbits included in $B$ and $C$.
%
%

Let $T\subseteqfin \A$ be an arbitrary finite subset of atoms. 
Both $B$ and $C$ split into setwise-$T$-orbits, refining~\eqref{eq:eqorbits}:
\begin{align} \label{eq:BC}
B = B_1 \uplus \ldots \uplus B_N \qquad
C = C_1 \uplus \ldots \uplus C_{M'}.
\end{align}
Let $C_1, \ldots, C_M$ be the \emph{finite} setwise-$T$-orbits among $C_1, \ldots, C_{M'}$
(clearly, $M$ and $N$ may depend on $T$).
Importantly, by Lemma \ref{lem:STorbit}, $N$ and $M$ do not depend on $T$ as long as $\size T\geq d$.
In fact $M = r$, the number of orbits included in $C$, as
by Lemma \ref{lem:STorbit} we deduce:
\begin{lemma} \label{lem:Tk}
Assuming $\size T \geq \ell$, the equivariant orbit $\otu {\A} \ell$ includes exactly one finite setwise-$T$-orbit, namely
$\otu T \ell$.
\end{lemma}
Our reduction proceeds in two steps:  first, 
we derive a finite
polynomially-parametrised system $P_1$, and then
we transform it further to a monotonic system $P_2$.
Monotonicity of $P_2$ guarantees correctness of reduction.

\para{Step 1 (finite polynomially-parametrised system)}
Our construction is parametric in $T$.
Let $b_1, \ldots, b_N$ be arbitrarily chosen representatives of setwise-$T$-orbits included in $B$.
Given $\vr A$ and $\vr t$,
we define an $N{\times}M$ matrix $\vr A_1(T)$ and a vector $\vr t_1(T) \in \Int^N$ as follows:
\begin{enumerate}
\item[(1)] 
Pick columns of $\vr A(T)$ indexed
by elements of all finite setwise-$T$-orbits included in $C$, and discard other columns; this yields a matrix 
$\vr A'(T)$ with finitely many columns (number thereof depending on $T$).
\item[(2)] 
Merge (sum up) columns of $\vr A'(T)$ indexed by elements of the same setwise-$T$-orbit; this yields a matrix $\vr A''(T)$ with $M$ columns ($M$ independent of $T$). 
\item [(3)]
Pick $N$ rows of $\vr A''$, indexed by $b_1, \ldots, b_N$,
and discard other rows; this yields an $N\times M$ matrix $\vr A_1(T)$.
\item[(4)] Likewise pick the corresponding entries of $\vr t$ and discard others, thus yielding a finite vector $\vr t_1(T)\in\Z^N$.
\end{enumerate}


\noindent
For $b\in B$ and $C_j \subseteq C$, $j \in \setto M$, 
we write $\orbsum {\vr A}(b, C_j)$ for the finite sum ranging over elements of $C_j$:
\[
\orbsum {\vr A}(b, C_j) = \sum_{c\in C_j} \vr A(b,c),
\]
%
which allows us to formally define the $B{\times} M$ matrix $\vr A''(T)$, 
the $N{\times}M$ matrix $\vr A_1(T)$ and the vector $\vr t_1(T) \in \Int^N$:
\begin{align} \label{eq:defAt}
\vr A''(T) (b, j) = \orbsum {\vr A}(b, C_j) 
\qquad\qquad
\vr A_1(T) (i, j) = \vr A''(T)(b_i, j) = \orbsum {\vr A}(b_i, C_j) 
\qquad\qquad
\vr t_1(T) (i) = \vr t(b_i).
\end{align}

\begin{slexample} \label{ex:cont}
We explain how the system \eqref{eq:ppar} in Example \ref{ex:poly} in Section \ref{sec:intro}
is obtained from
the system \eqref{eq:matr} in Example \ref{ex:lp}.
Fix a non-empty $T \subseteqfin \A$. 
The set $\A$ includes just one finite setwise-$T$-orbit, namely $T$.
Therefore the matrix $\vr A'(T)$ has $\size T$ columns, $\vr A''(T)$ has just one column,
and the system $(\vr A_1(T), \vr t_1(T))$ has just one unknown.
Furthermore, the set $\A$ includes two setwise-$T$-orbits, the finite one $T$ plus the
infinite one $\A \setminus T$,
and therefore the system $(\vr A_1(T), \vr t_1(T))$ has two inequalities.
Pick arbitrary representatives of the setwise-$T$-orbits, $b_1 \in T$ and $b_2 \in (\A\setminus T)$.
We have
\begin{align*}
\vr A_1(T)(1,1) \  = \ \sum_{c\in T }\vr{A}(b_1, c) = |T| - 1 
\qquad\qquad
\vr A_1(T)(2,1) \  = \  \sum_{c\in T}\vr{A}(b_2, c) = |T|.
\end{align*}
Replacing $|T|$ with $n$ yields the system $(\vr A_1(T), \vr t_1(T))$: 
\begin{align} \label{eq:pparnmon}
\begin{bmatrix}
n-1 \\
n
\end{bmatrix}
\cdot
x
\ \geq \ 
\begin{bmatrix}
1 \\
1
\end{bmatrix}
\end{align}
which happens to be monotonic.
In general, the system obtained so far needs not be monotonic,
but we will ensure monotonicity in the subsequent step.
\end{slexample}

The choice of representatives $b_i$ is irrelevant, and hence
$\vr A_1(T)$ and  $\vr t_1(T)$ are well defined, 
since rows of $\vr A''$ indexed by any two elements of $B$ belonging the same setwise-$T$-orbit
are equal, and likewise the corresponding entries of $\vr t$:
%
\begin{lemma}\label{lem:equalised row}
If $b, b' \in B$ are in the same setwise-$T$-orbit, then 
$\vr t(b) = \vr t(b')$ and
$\orbsum {\vr A}(b, C_j) = \orbsum {\vr A}(b', C_j)$
for every $j\in \setto M$.
\end{lemma}
\begin{proof}
Let $\pi \in \Aut{(T)}$ be such that $\pi(b) = b'$.
As $\vr t$ is equivariant, it is necessarily constant on the whole equivariant orbit to which $b$ and $b'$ belong
(cf.~Lemma \ref{lem:1O}), 
and hence $\vr t(b') = \vr{t}(b)$.

For the second point fix $j\in \setto M$.
As $\vr A$ is equivariant, it is constant over the orbit included in $B\times C$ to which $(b,c)$ belongs, 
for every $c\in C$, and hence
$\vr A(b,c) = \vr A(\pi(b), \pi(c))$. This implies
\[
\sum_{c \in C_j} \vr{A}(b,c) = 
\sum_{c \in C_j} \vr A(\pi(b),\pi(c)) = 
\sum_{c \in C_j} \vr A(b',\pi(c)).
\]
Since $\pi$ is a setwise-$T$-automorphism, when restricted to
the setwise-$T$-orbit $C_j$ it is a bijection $C_j \to C_j$,
and hence the two sums below differ only by the order of summation and are thus equal:
\[
\sum_{c \in C_j} \vr A(b',\pi(c)) = 
\sum_{c \in C_j} \vr A(b',c).
\]
The two above equalities imply the claim, namely
$
\sum_{c \in C_j} \vr A(b,c) = 
\sum_{c \in C_j} \vr A(b',c).
$ 
\end{proof}

\begin{notation} \label{not:M}
Let $T\subseteqfin \A$.
Due to Lemma \ref{lem:Tk},
the set of finite setwise-$T$-orbits $\set{C_1, \ldots, C_M}$ included in $C$
is in bijection with 
the set $\orbits C = \set{\O_1, \ldots, \O_M}$ of equivariant orbits included in $C$.
W.l.o.g.~assume $C_j\subseteq \O_j$ for $j = 1\ldots M$.
Take any finitary setwise-$T$-supported vector $\vr x : C \tofin \R$.
It is non-zero only inside finite setwise-$T$-orbits $C_j$, which implies
\[
\orbsum {\vr x}(C_j) \ = \ \orbsum {\vr x}(\O_j)
\]
for $j=1\ldots M$  (cf.~Notation \ref{not:orbsum}).
Furthermore, $\vr x$ is constant inside each $C_j$, which allows us to write
$\orbval {\vr x}(C_j)$ (cf.~Notation \ref{not:orbval}).
For notational convenience we slightly relax Notations \ref{not:orbval} and \ref{not:orbsum}
from now on, and treat the orbit-value and orbit-sum vectors as $M$-tuples,
$
\orbval{\vr x}, \ \orbsum{\vr x} \ \in \ \R^M,
$
with obvious meaning $\orbval{\vr x}(j) = \orbval{\vr x}(C_j)$ and $\orbsum{\vr x}(j) = \orbsum{\vr x}(C_j)$.
We note the (obvious) relation between $\orbval{\vr x}$ and $\orbsum{\vr x}$:
\begin{align} \label{eq:overbefore}
\orbsum{\vr x}(j) \ = \ \size{C_j} \cdot \orbval{\vr x}(j).
\end{align}
\end{notation}

The following lemma, being a cornerstone of correctness of the whole reduction,
is now not difficult to prove:
\begin{lemma} \label{lem:givenT}
Let $\size T \geq d$ and $\vr x : C\tofin \R$ a finitary setwise-$T$-supported vector.
The following conditions are equivalent:
\begin{itemize}
\item 
$\vr x$ is solution of $(\vr A, \vr t)$;
\item
$\orbval {\vr x}$ is a solution of $P_1(T) = (\vr A_1(T), \vr t_1(T))$.
\end{itemize}
\end{lemma}
%
%
\begin{proof}
Take any setwise-$T$-supported vector $\vr x : C\tofin \R$, and
let $\vr x'$ be the restriction of $\vr x$ to $C' = C_1\uplus\ldots\uplus C_M$.
We argue that the following four conditions are equivalent, which implies the claim:

\begin{enumerate}
\item 
$\vr x$ is solution of $(\vr A, \vr t)$;
\item 
$\vr x'$ is solution of $(\vr A'(T), \vr t)$;
\item 
$\orbval{\vr x}$ is solution of $(\vr A''(T), \vr t)$;
\item
$\orbval {\vr x}$ is a solution of $(\vr A_1(T), \vr t_1(T))$.
\end{enumerate}

First, as $\vr x$ is finitary,
we have $\vr x(c)=0$ for all $c\notin C'$, and hence 
$\vr{A}'(T) \cdot \vr{x}' = \vr{A} \cdot \vr{x}$.  
This implies equivalence of (1) and (2).
Second, as $\vr A''$ is obtained from $\vr A'$ by summing columns over a setwise-$T$-orbit 
where the vector $\vr x$, being setwise-$T$-supported, is constant,
we have $\vr{A}''(T) \cdot \orbval{\vr{x}} = \vr{A}'(T) \cdot \vr{x}'$.
This implies equivalence of (2) and (3).
Finally, (3) implies (4) as $(\vr A_1(T), \vr t_1(T)$ is obtained from $(\vr A''(T), \vr t)$ by removing 
inequalities.
For the reverse implication, we recall that Lemma \ref{lem:equalised row} shows that
$\vr A''(b, j) = \vr A''(b_i, j)$ and $\vr t(b) = \vr t(b_i)$ for every $i\in\setto N$ and $b\in B_i$, 
and therefore $\vr A''(T)$ contains the same inequalities as $\vr A_1(T)$.
In consequence, (4) implies (3).
\end{proof}
The function $T\mapsto P_1(T)$ is equivariant, i.e., invariant under action of atom automorphisms.
In consequence, the entries of $\vr A_1(T)$ and $\vr t_1(T)$ do not depend on the set $T$ itself, but only
on its size $\size T$. 
Indeed, if $\size T = \size{T'}$ then $\pi(T) = T'$ for some atom automorphism $\pi$, and hence
$\pi(P_1(T)) = P_1(T')$.
Since the system $P_1(T)$ is atom-less 
we have also $\pi(P_1(T)) = P_1(T)$,
which implies 
$ P_1(T) = P_1(T'). $
We may thus meaningfully write $P_1(\size T) = (\vr A(\size T), \vr t(\size T))$, i.e.,
$P_1(n) = (\vr A_1(n), \vr t_1(n))$ for $n\in\Nat$
 (cf.~Example \ref{ex:cont}).
%


\smallskip

We argue that the dependence on $\size T$ is polynomial, as long as $\size T\geq 2d$:
\begin{lemma} \label{lem:poly}
There are univariate polynomials $p_{ij}(n) \in \Int[n]$ 
such that
$\vr A_1(n)(i,j) = p_{ij}(n)$ for $n \geq 2d$.
\end{lemma}
%
%
\begin{proof} 
Let $n=\size T$.
Fix a setwise-$T$-orbits $B_i\subseteq B$ and a finite setwise-$T$-orbit $C_j\subseteq C$.
Each of them is included in a unique equivariant orbit, say:
\[
B_i \subseteq B' =  \otu {\A} p \qquad
C_j \subseteq C' = \otu {\A} \ell
\] 
(cf.~the partitions~\eqref{eq:eqorbits}).
Recall Lemma~\ref{lem:STorbit}:  $B_i$ is determined by the subset 
$I \subseteq \setto p$ of positions where atoms of $T$ appear in tuples belonging to $B_i$.
Let $m = \size{I}$.
On the other hand $C_j = \otu T \ell$ (cf.~Lemma~\ref{lem:Tk}).
%
%
%
Note that $m = \size{T\cap \supp{b_i}}$.

\smallskip

We are going to demonstrate that the value $\orbsum{\vr A}(b_i, C_j)$ is
polynomially depending on $n=\size T$.
We will use the polynomials 
$\otun n w$
of degree $w$, for $w\leq d$, defined by 
\begin{align} \label{eq:poly2}
\begin{aligned}
\otun n w & = n \cdot (n-1) \cdot \ldots \cdot (n-w+1).
\end{aligned}
\end{align}
In the special case of $w=0$, we put $\otun n w = 1$.
The value $\otun n w$ can be interpreted as follows:
\begin{claim} \label{claim:polyint}
For $n\geq w$, $\otun n w$ is equal to the number of 
arrangements of $w$ items chosen from  $n$ objects into a sequence.
\end{claim}

Denote by $\mathcal D$ the set of equivariant orbits $\O \subseteq B' \times C'$.
For $\O\in \mathcal D$, we put $\O(b_i, C_j) := \setof{c\in C_j}{(b_i,c)\in \O}$.
%
%
As $\vr A$ is equivariant, the value $\vr A(b_i, c)$ depends only on the orbit to which $(b_i, c)$ belongs.
We write $\vr A(\O)$, for $\O\in\mathcal D$, and get:
\begin{claim} \label{claim:sum}
$\orbsum {\vr A}(b_i, C_j) = \sum_{\O\in \mathcal D} \vr A(\O) \cdot \size {\O(b_i, C_j)}$.
\end{claim}

By Lemma \ref{lem:orbitpair} in Section \ref{sec:atoms}, orbits $\O\subseteq B' \times C'$ are in 
one-to-one correspondence
with partial injections $\iota : \setto p \to \setto \ell$.
We write $\O_\iota$ for the orbit corresponding to $\iota$.
Let $\dom \iota = \setof{x}{\iota(x) \text{ is defined}}$ denote the domain of $\iota$.
\begin{claim} \label{claim:Uiotaempty}
$U_\iota(b_i, C_j)\neq \emptyset$ if and only if $\dom{\iota}\subseteq I$. 
\end{claim}
\noindent
Indeed, recall again Lemma \ref{lem:orbitpair}  which yields
$
\O_\iota(b_i, C_j) = \setof{ c \in C_j }{\prettyforall{x, y}{b_i(x) = c(y) \iff \iota(x) = y}}.
$
If $\dom{\iota}\subseteq I$, the set $\O_\iota(b_i, C_j)$ contains tuples $c\in C_j$ with fixed values
on positions $J = \setof{\iota(x)}{x\in\dom \iota}$, namely
\begin{align} \label{eq:iotam}
b_i(x) = c(\iota(x)),
\end{align}
and arbitrary other atoms from $T$ elsewhere, and therefore is nonempty.
If there is $x \in \dom{\iota}\setminus I$ then $b_i(x)\notin T$ and therefore
no $c\in C_j$ satisfies \eqref{eq:iotam}.
Claim \ref{claim:Uiotaempty} is thus proved.
\begin{claim} \label{claim:polOC}
Let $k = \size{\dom{\iota}}$ be
the number of pairs related by $\iota$.
If $\O_\iota(b_i, C_j)\neq \emptyset$ then $\size {\O_\iota(b_i, C_j)}  = \otun {(n-m)} {\ell - k}$. 
\end{claim}
\noindent
According to \eqref{eq:iotam}, tuples $c\in \O_\iota(b_i)$ have fixed values on $k$ positions in $J$.
The remaining $\ell - k$ positions in tuples $c\in \O_\iota(b_i,C_j)$ are filled arbitrarily
using $n-m$ atoms from $T\setminus \supp{b_i}$.
Due to the assumption that $n\geq 2d$, we have $n-m \geq d \geq  \ell - k$, and therefore
using Claim \ref{claim:polyint} (for $w=\ell-k$) we deduce 
$\size {\O_\iota(b_i,C_j)} = \otun {(n-m)} {\ell-k}$,  
thus proving Claim \ref{claim:polOC}.

\smallskip

Once $b_i\in B_i$  and $\O\in\mathcal D$  
are fixed, the values $k, \ell$ and $m$ are fixed too, and
the formula of Claim~\ref{claim:polOC}
is an univariate polynomial of degree $\ell-k$.
%
%
%
%
The formula of  Claim~\ref{claim:sum} yields the required polynomial%
\footnote{
This confirms, in particular,  that $\vr A(T)(i,j)$ is independent from the actual set $T$, 
and only depend on its size $n=\size T$.}
 $\vr A(n)(i,j) = p_{ij}(n)$ and hence the proof of Lemma~\ref{lem:poly} is completed.
\end{proof}

Relying on Lemma \ref{lem:poly} we get a polynomially-parametrised system 
$P_1(n) = (\vr A_1(n), \vr t_1(n))$.

\para{Step 2 (monotonicity)}
The system $P_1$ constructed so far, does \emph{not} have to be monotonic in general.
As an immediate corollary of Lemma~\ref{lem:givenT} and Corollary \ref{cor:weaksol}, we only know that
If $P_1(n)$ has a solution for $n\geq d$, then
$P_1(n+1)$ has a (potentially different) solution.
We slightly modify the system $P_1$
in order to achieve monotonicity.

Before defining formally the new system 
$P_2(n) = ({\vr A_2}(n), {\vr t_2}(n))$,
we point to our objective: we aim at replacing the orbit-value vector $\orbval {\vr x}$ 
in Lemma \ref{lem:givenT} by the orbit-sum vector
$\orbsum {\vr x}$, as in Lemma \ref{lem:givenTT} below.
In other terms, we want the solutions $\vr y_1$ and ${\vr y_2}$ of $P_1(n)$ and
$P_2(n)$, respectively, 
%
differ on position $j$ by the multiplicative factor of $\size{C_j}$, 
namely
\begin{align} \label{eq:over}
{\vr y_2}(j) \ = \ \size{C_j} \cdot \vr y_1(j)
\end{align}
for $j = 1, \ldots, M$ (cf.~\eqref{eq:overbefore}).
The size $\size {C_j}$ of the setwise-$T$-orbit $C_j$, where $\size T = n$, is equal to
\begin{align} \label{eq:orbitsize}
\size{C_j} \ = \ \otun n {e_j},   
\end{align}
where $C_j \subseteq \otu \A {e_j}$, i.e., 
$e_j$ is the atom dimension of the equivariant orbit including $C_j$, assuming
$\size T \geq e_j$.
These considerations lead to the following formal definition of $P_2$:
\begin{align} \label{eq:defP}
\begin{aligned}
{\vr A_2}(n)(i,j) \ & = \  \vr A_1(n)(i,j) \cdot
\frac{\otun n d}{\otun n {e_j}} 
\qquad\qquad\qquad
{\vr t_2}(i) \ & = \ \vr t_1(i) \cdot \otun n d  
\end{aligned}
\end{align}
where $\vr A_1(n)(i,j) = p_{ij}(n)$.
We rely on the following fact:
\begin{claim} \label{claim:qeq}
$\otun n w \cdot \otun {(n-w)} u \ = \ \otun n {w+u}$. 
\end{claim}
\noindent
By the claim, all coefficients in \eqref{eq:defP} are polynomials, namely:
$
{\vr A_2}(n)(i,j) \ = \  p_{ij}(n) \cdot \otun {(n-e_j)} {d-e_j},  
$
since $e_j \leq d$.
It remains to conclude that the systems $P_1(n)$ and $P_2(n)$ have the same 
solutions modulo \eqref{eq:over}:
\begin{lemma} \label{lem:modulo}
Let $n\geq d$ and let 
$\vr y_1, {\vr y_2}\in\R^M$ satisfy 
$
{\vr y_2}(j) \ = \  \otun n {e_j} 
\cdot \vr y_1(j)
$
for $j = 1, \ldots, M$.
Then
$\vr y_1$ is a solution of $P_1(n)$ if and only if \xspace ${\vr y_2}$ is a solution of $P_2(n)$.
\end{lemma}
Combination of \eqref{eq:overbefore}, \eqref{eq:orbitsize},
and Lemmas \ref{lem:givenT} and \ref{lem:modulo} yields:
\begin{lemma} \label{lem:givenTT}
Let $\size T = n \geq d$ and $\vr x : C\tofin \R$ a finitary setwise-$T$-supported vector.
The following conditions are equivalent:
\begin{itemize}
\item 
$\vr x$ is solution of $(\vr A, \vr t)$;
\item
$\orbsum {\vr x}$ is a solution of $P_2(n)$.
\end{itemize}
\end{lemma}
\begin{slexample} \label{ex:contcont}
The system \eqref{eq:ppar} in Example \ref{ex:poly} is obtained from \eqref{eq:pparnmon}
in Example \ref{ex:cont} by applying the definition \eqref{eq:defP}.
Indeed, $M = d = e_1 = 1$ and hence $\vr A$ stays unchanged, while
the right-hand side vector $\vr t$ gets multiplied by $\otun n 1  = n$.
\end{slexample}


\begin{lemma}[Monotonicity] \label{lem:monoto}
Let $n\geq d$. Every solution of $P_2(n)$ is also a solution 
of $P_2(n+1)$.
\end{lemma}
\begin{proof}
Suppose ${\vr y}$ is a solution of $P_2(n)$, for $n\geq d$.
Let $T\subseteqfin\A$ be any subset of atoms of size $\size T = n$.
Let $\vr x$ be the finitary setwise-$T$-supported vector uniquely determined by
\begin{align} \label{eq:overyone}
\orbsum {\vr x} \ = \  {\vr y}.
\end{align}
By Lemma \ref{lem:givenTT}, $\vr x$ is a solution of $(\vr A, \vr t)$.
We apply Corollary \ref{cor:weaksol} to obtain another finitary solution 
$\vr x'$ of $(\vr A, \vr t)$, 
setwise-$T'$-supported by some $T'$ of size $\size T' = n+1$,
and having the same orbit-summation mapping:
\begin{align}\label{eq:sums}
\orbsum {\vr x} \ = \  \orbsum {(\vr x')}.
\end{align}
Equalities \eqref{eq:overyone} and \eqref{eq:sums} imply $ \orbsum {(\vr x')} = {\vr y}$.
By Lemma \ref{lem:givenTT} again, $ \orbsum {(\vr x')} = {\vr y}$
is a solution of $P_2(n+1)$, as required.
\end{proof}

%
%
%


Combining Lemmas \ref{lem:weaksol}, \ref{lem:poly}, \ref{lem:givenTT} and \ref{lem:monoto}
we derive correctness of reduction
(the constraint $n\geq 2d$ is inherited from the assumption in Lemma \ref{lem:poly}):
\begin{corollary} \label{cor:iff}
The following conditions are equivalent:
\begin{itemize}
\item
$(\vr A, \vr t)$ has a finitary solution,
\item
$P_2(n)$ has a solution for some integer $n \geq 2d$,
\item
$P_2(n)$ has a solution for almost all $n \in \Nat$.
\end{itemize}
\end{corollary}
%
%
%
%
%

Reduction of \finineqsolvname$(\R)$ to \allpolyineqsolvname is thus completed.

\para{Complexity}
It remains to argue that $P_2$ is computable from $(\vr A, \vr t)$, and estimate the 
computational complexity.
Computability of $P_2$ follows immediately from computability of $P_1$, which we focus now on:
\begin{lemma} \label{lem:comp}
The system $P_1$ is computable from $(\vr A, \vr t)$. 
\end{lemma}
\begin{proof}
Indeed, it is enough to range over representations of setwise-$T$-orbits $B_i$ and $C_j$ of $B$ and $C$, respectively
(such representations are given by Lemma~\ref{lem:STorbit}), and for each
pair of such orbits proceed with computations outlined in the proof of Lemma~\ref{lem:poly}, applied
to an arbitrarily chosen representative $b_i\in B_i$. 
\end{proof}

By Corollary \ref{cor:iff} and Lemma \ref{lem:comp},
\finineqsolvname$(\R)$  reduces to \allpolyineqsolvname.

\smallskip

%




Concerning computational complexity,
the number of setwise-$T$-orbits included in an equivariant orbit $\otu \A \ell$ is exponential in $\ell$
(cf.~Lemma \ref{lem:STorbit}).
That is why 
the size of  $P_2$ may be exponential in atom dimension $d$ of $(\vr A, \vr t)$.
On the other hand, the size of $P_2$ is only polynomial
(actually, linear) in the number of orbits included in $B\times C$.
In consequence, for fixed atom dimension we get a polynomial-time reduction and hence,
relying on Theorem \ref{thm:polyptime},
the decision procedure for \finineqsolvname$(\R)$ in \PTIME.
Without fixing atom dimension, we get an exponential-time reduction and hence
the decision procedure is in \EXPTIME.

The same complexity bounds apply to the algorithm for the optimisation problem presented in Section \ref{sec:max}.



\section{Optimisation problems}
\label{sec:max}

In this section we prove Theorem~\ref{thm:lp-max}:
we introduce a maximisation variant of \polyineqsolvname and routinely adapt
the decision procedure of Section \ref{sec:poly}, as well as
the reduction of Section \ref{sec:red2poly},
to the maximisation setting.


\subsection{Polynomially-parametrised maximisation problem}

We consider a maximisation problem, whose instance 
$(P,S)$ consists of a finite system of polynomially-parametrised 
inequalities $P$ as in \eqref{eq:polyeq}, and an \emph{ordinary} (non-parametrised)
objective function $S$ given by a linear map
\begin{align*} 
S(x_1, \ldots, x_k) = 
a_1 \cdot x_1 + \ldots + a_k \cdot x_k.
\end{align*}
%
%
%
%
As in Section \ref{sec:poly-mon}, 
by an \emph{\aasol} of a system $P$ we mean in this section a solution of $P(n)$ 
for almost all $n\in\Nat$.  
We define the \emph{supremum} of a monotonic instance $(P, S)$ 
as
\[
\sup (P, S) \ := \ \sup \setof{S(\vr x)}{\vr x \text{ is an \aasol of } P},
\]
under a proviso that $\sup (P,S) = -\infty$ if $P$ has no {\aasol}s.
Referring to the standard terminology, we can say that the system is \emph{infeasible} if
$sup (P,S) = -\infty$, and it is \emph{unbounded} if $\sup (P, S) = \infty$.
%
%
Interestingly, the supremum can not be irrational
(see Corollary \ref{cor:limQ} below).
%


In this section we study the problem of computing the supremum of monotonic instances:

\compprobext{\allpolyineqmaxname} 
{An instance $(P, S)$}
{The supremum of $(P, S)$}

\noindent
The problem generalises ordinary (non-parametrised) linear programming, and can be solved
similarly to \allpolyineqsolvname (of which it is a strengthening):
%
\begin{theorem} \label{thm:lp-max-poly}
\allpolyineqmaxname is in \PTIME.
\end{theorem}
\begin{proof}
Let $(P_0, S)$ be an instance.
The algorithm is essentially the same as Algorithm \ref{alg:1} for \allpolyineqsolvname 
in the proof of Theorem \ref{thm:polyptime}, and proceeds
by iterating the transformation step until either unsolvability or solvability is reported.
Recall that solution set is preserved by the transformation step (Claim \ref{claim:samesol}).
If unsolvability is reported, the algorithm returns $-\infty$.
If solvability is reported --- let $P\cup \G$ be the system examined in the last iteration
---
the decision procedure computes and returns $\sup (\li P \conj \G, S)$, 
the supremum of $S$ constrained by the ordinary system of inequalities $\li  P \conj \G$,
by invoking any \PTIME procedure for ordinary linear programming.

Correctness follows by the two claims formulated below. 
First, since solution set is preserved by the transformation step, we have:
\begin{claim} \label{claim:proof1}
$\sup (P_0, S) = \sup (P \cup \G, S)$.
\end{claim}
\noindent
Second, the supremum does not change if
the polynomially-parametrised constraints $P$ are replaced by the 
overapproximation $\li P$:
\begin{claim} \label{claim:proof2}
$\sup (P\cup\G, S) =  \sup (\li P \cup \G, S)$. 
\end{claim}
\noindent
%
For the claim it is enough to prove the inequality
$
\sup (\strictli P\cup\G, S) \geq  \sup (\li P \cup \G, S)
$
as, according to Claims \ref{claim:li} and \ref{claim:strictli}, we have
$\sup (\strictli P \cup \G, S) \leq \sup (P\cup\G, S) \leq  \sup (\li P \cup \G, S)$.
Take any solution $\vr y$ of $\li P\cup\G$, and any solution 
$\vr x$ of $\strictli P\cup\G$ (we rely here on solvability of the latter system).
For every $k\in\Nat$, the vector
$\vr x_k = \frac{\vr x + k \vr y}{k+1}$ is a solution of $\strictli P\cup\G$, and $S(\vr x_k)$
tends to $S(\vr y)$ when $k$ tends to $\infty$.
Hence 
$
\sup (\strictli P\cup\G, S) \geq  \sup (\li P \cup \G, S),
$
as required.
%
%
\end{proof}

By Claims \ref{claim:proof1} and \ref{claim:proof2} in the proof of Theorem \ref{thm:lp-max-poly},
the supremum of a monotonic instance, if not  $-\infty$ nor $\infty$,
is equal to the supremum of an ordinary linear program and hence is rational:
\begin{corollary} \label{cor:limQ}
The supremum of a monotonic instance $(P,S)$ belongs to $\Rat\cup\set{-\infty,+\infty}$. 
\end{corollary}

\begin{slremark}
As illustrated in Example \ref{ex:poly} in Section \ref{sec:intro}, 
the objective function $S$ may not achieve its supremum over the {\aasol}s of $P$.
Once the supremum $s\in \Rat$ is computed, 
one can easily check if $S$ achieves its supremum, by adding to the system an equation
$S(x_1, \ldots, x_k) = s$ and checking if the system is still solvable.
\end{slremark}

\subsection{Reduction of \finineqmaxname$(\R)$ to  \allpolyineqmaxname}


We only sketch the reduction as it amounts to a slight adaptation of the reduction
 of Section \ref{sec:red2poly}.
The input of \finineqmaxname$(\R)$ consists of a system $(\vr A, \vr t)$ and 
an integer vector $\vr s : C \tofs \Z$ representing
the objective function
\[
S(\vr x) = \innerprod{\vr s}{\vr x},
\]
and we ask for the supremum of values $S(\vr x)$, for $\vr x$ ranging over finitary solutions of $(\vr A, \vr t)$.
This value we denote as $\sup(\vr A, \vr t, \vr s)$.
%
In addition to Lemma \ref{lem:wlog} we show (the proof is in Section \ref{sec:lp-decid-proof-proofs}):
%
\begin{lemma} \label{lem:wlogmax}
W.l.o.g.~we may assume that $\vr s$ is equivariant.
\end{lemma}
We proceed by adapting the reduction of Section \ref{sec:red2poly}:
given an instance $(\vr A, \vr t, \vr s)$ of \finineqmaxname$(\R)$ we compute a monotonic instance
 $(P_2, S')$ of \polyineqmaxname, where
the finite system $P_2(n) = (\vr A_2(n), \vr t_2(n))$ of polynomially-parametrised inequalities 
is exactly as in Section \ref{sec:red2poly}, and
the objective function is
\begin{align} \label{eq:newobj}
S'(x_1, \ldots, x_k) = a_1 \cdot x_1 + \ldots + a_M \cdot x_M,
\end{align}
where $a_j = \orbval{\vr s}(C_j)$ for $j=1\ldots M$
(recall Notations \ref{not:orbval} and \ref{not:M}).
More concisely, the vector $\vr a = a_1 \ldots a_M$ is defined as
$\vr a = \orbval {\vr s}$.
%
%
%
We apply Lemmas \ref{lem:weaksol} and \ref{lem:givenTT}
to obtain:
\begin{lemma}  \label{lem:givenTobj}
$\sup (\vr A, \vr t, \vr s) = \sup (P_2, S')$.
\end{lemma}
\begin{proof}
Let $\vr x : C\tofin\R$. By equivariance of $S$ and the definition of $S'$ we have the equality
$
S(\vr x) \ = \ S'(\orbsum{\vr x}),
$
that is,
the value of the objective function $S(\vr x)$ depends only on the orbit-sum vector 
$\orbsum{\vr x} : \orbits C \to \R$.
As Lemmas \ref{lem:weaksol} and \ref{lem:givenTT} preserve orbit-sum, we deduce that
for every $T\subseteqfin\A$ of size $\size T =n \geq 2d$,
the values of $S$ on finitary $T$-supported solutions of $(\vr A, \vr t)$ are the same as
the values of $S'$ on solutions of $P_2(n)$.
By Lemma \ref{lem:monoto}, the solutions of $P_2(n)$ for some $n\geq 2d$ are exactly the same as
the {\aasol}s of $P_2$.
In consequence, the two suprema are equal.
%
%
%
%
\end{proof}
%


\begin{slexample}
To illustrate the reduction, consider the modification of the system in Example \ref{ex:Kirchoff}:
%
\begin{align}  \label{eq:K22}
\sum_{\a \in \A} \a \ \geq \  1
\qquad\qquad\qquad\qquad
\sum_{\b \in \A} \a\b \ - \ \a \ - \ 2 \cdot \sum_{\b \in \A} \b\a  \ \geq \ 0 \qquad\qquad (\a \in\A).
\end{align}
It enforces, for each vertex $\a\in\A$, the sum of values assigned to all outgoing edges
to be larger than \emph{double} the sum of values assigned to all ingoing edges, plus the value assigned
to the vertex $\a$.
The indexing sets $B=\A\cup\set{*}$ and $C=\A\cup\otu \A 2$ 
and the shape of the matrix \eqref{eq:Kirchoff-matr} are the same.
We identify the singletons $\set{*} = \otu \A 0$.
We consider maximisation of triple the sum of values assigned to edges:
$S(\vr x) = \innerprod {\vr s}{\vr x}$, where $\vr s = 3\cdot \constvr 1 {\otu \A 2}$, or
$
S(\vr x) \ = \ 3 \cdot
\sum_{\a\b\in\otu\A 2} \vr x (\a\b).
$ 

According to Lemma \ref{lem:Tk},
the set $C$ includes exactly 2 finite setwise-$T$-orbits, namely
$T\subseteq \A$ and $\otu T 2 \subseteq \otu \A 2$, and therefore
the system computed by the reduction has 2 unknowns, $x_1$ and $x_2$.
By Lemma \ref{lem:STorbit}, 
for any nonempty $T \subseteqfin \A$,
the set $B$ includes 3 setwise-$T$-orbits,
namely $T$, $\A\setminus T$ and $\set{*}$,
and therefore the system $P_1$ computed in the first step has 3 inequalities:
\begin{align}
\label{eq:ukl}
\begin{aligned}
-x_1 \ - \, (n-1) \cdot x_2 & \ \geq \  0 \qquad\qquad && (T) \\
0 & \ \geq \  0 && (\A\setminus T) \\
n\cdot x_1  & \ \geq \ 1 && (\set{*})
\end{aligned}
\end{align}
%
%
%
For instance, the coefficient $-(n-1)$ in the first inequality arises as:
\begin{align*}
 \vr A(U_\text{out}) \cdot \size{U_\text{out}(\a, \otu T 2)} \ + \
\vr A(U_\text{in}) \cdot  \size{U_\text{in}(\a, \otu T 2)}  \  = \  
 1 \cdot (n-1) -2 \cdot (n-1) \ = \ -(n-1)
\end{align*}
(cf.~Claim \ref{claim:sum}),
for some arbitrary $\a\in T$  and the following two orbits included in $\A \times \otu \A 2$:
\begin{align*}
& U_\text{out} = \setof{\!(\a, \a\b)\!}{\!\b\neq\a\!},
\ \ 
U_\text{in} = \setof{\!(\a, \b\a)\!}{\!\b\neq\a\!}.
\end{align*}
%
%
Likewise, the coefficient $n$ in the last inequality arises as
$
\vr A(U) \cdot \size{O(*, T)} = 1 \cdot n = n,
$
for the orbit
$
U = \set{*}\times\A.
$
According to \eqref{eq:defP},
the system $P_2$ is obtained from \eqref{eq:ukl} by multiplying all occurrences of
$x_1$ by $\otun {(n-1)} 1  = n-1$, and by multiplying all right-hand sides by
$\otun n 2  = n(n-1)$ (the trivial second inequality is omitted):
\begin{align}
\label{eq:ukltwo}
\begin{aligned}
-(n-1)\cdot x_1 \ - \, (n-1) \cdot x_2 & \ \geq \  0  \\
n(n-1) \cdot x_1  & \ \geq \ n(n-1)
\end{aligned}
\end{align}
Finally, the objective function produced by the reduction, as in \eqref{eq:newobj}, is
$
S'(x_1, x_2) \ = \ \vr s({\otu \A 2}) \cdot x_2 \ = \   3 \cdot x_2.
$
It achieves $-3$ as its supremum,
as the system \eqref{eq:ukltwo} is equivalent to the ordinary system (its head):
\[
x_1 \geq 1 \qquad\qquad x_2 \leq - x_1.
\]
%
For every $n\geq 2$, 
the optimal solution $x_1 = 1$, $x_2 = -1$ corresponds, via the constructions of Section \ref{sec:red2poly},
to a clique of $n$ vertices where each vertex is assigned $\frac 1 n$, 
and each edge is assigned 
$-\frac 1 {n(n-1)}$.
\end{slexample}


\section{Undecidability of integer solvability} \label{sec:ilp-undecid-proof}

\newcommand{\zeroinstr}{\text{\sc zero}}
\newcommand{\updinstr}{\text{\sc upd}}
We prove Theorem~\ref{thm:ilp-undecid} by showing undecidability of \finineqsolvname$(\Int)$.
We proceed by reduction from the reachability problem of counter machines.

We conveniently define a \emph{$d$-counter machine}
as a finite set of instructions $I$, where each instruction is a function
\[
i : \set{1 \ldots d} \to \Int \cup \set{\zeroinstr}
\]
that specifies, for each counter $k\in\set{1,\ldots, d}$, either the additive update of $k$ (if $i(k) \in \Int$)
or the zero-test of $k$ (if $i(k) = \zeroinstr$).
Configurations of $M$ are nonnegative vectors $c\in \Nat^d$, and
each instruction induces steps between configurations: $c \trans {i} c'$ if
$c'(k) = c(k) + i(k)$ whenever $i(k)\in\Nat$, and $c'(k) = c(k) = 0$ whenever $i(k)=\zeroinstr$.
A run of $M$ is defined as a finite sequence of steps
\begin{align} \label{eq:run}
c_0 \trans{i_1} c_1 \trans{i_2} \ldots \trans{i_n} c_n .
\end{align}
The reachability problem asks, given a machine $M$ and two its configurations, a source $c_0$ and
a target $c_f$, if $M$ admits a run from $c_0$ to $c_f$.
The problem is undecidable, as counter machines can easily simulate classical Minsky machines.%
\footnote{A $d$-counter machine resembles a vector addition system with zero tests.
A Minsky machine with $n$ states and $k$ counters can be simulated by an $(n+k)$-counter machine,
by encoding control states into additional counters.
} 

For $k\in\set{1,\ldots,d}$ we denote by $\zeroinstr(k) = \setof{i\in I}{i(k)=\zeroinstr}$ 
the set of instructions that zero-test counter $k$, and
$\updinstr(k) = \setof{i\in I}{i(k)\in\Int}$ 
the set of instructions that update counter $k$.

Given a $d$-counter machine $M$ and two configurations $c_0, c_f$, 
we construct an orbit-finite system of inequalities $S = (\vr A, \vr t)$ 
such that $M$ admits a run from $c_0$ to $c_f$ if and only if 
$S$ has a finitary nonnegative integer solution.
(Nonegativeness is \emph{enforced} by adding inequalities  $x\geq 0$ for all unknowns $x$.)
We describe construction of $S$ gradually, on the way giving intuitive explanations 
and sketching the proof of the if direction.

The system $S$ has unknowns $e_{\a \b}$ indexed by pairs of distinct atoms $\a\b\in\otu {\A} 2$, and
contains the following inequalities:
\begin{align}\label{eq:graph constraint}
e_{\a \b} &\leq 1 \qquad\qquad (\a\b\in\otu \A 2).
\end{align}
Therefore, in every solution the unknowns $e_{\a\b}$ define a directed graph $G$, where atoms are vertices, 
$e_{\a\b}=1$ encodes an edge from $\a$ to $\b$ and $e_{\a\b}=0$ encodes a non-edge.
In case of a finitary solution, the graph $G$ is finite (when atoms with no adjacent edges are dropped).
Let us fix two distinct atoms $\iota, \zeta\in \A$.
The system $S$ contains the following further equations and inequalities:
\begin{align}\label{eq:degree constraint one}
\sum_{\b\neq \a} e_{\b\a} \ =  \ \sum_{\b\neq \a} e_{\a\b} \leq 1 \qquad\qquad
(\a\in\A\setminus\set{\iota,\zeta}) 
\end{align}
enforcing that in-degree of every vertex, except for $\iota$ and $\zeta$, is the same as its out-degree, and equal 0 or 1, and also
\begin{align}\label{eq:degree constraint two}
\sum_{\b\neq \iota} e_{\b\iota} = 0 \qquad\qquad \sum_{\b\neq \iota} e_{\iota\b} = 1 \qquad\qquad
\sum_{\b\neq \zeta} e_{\b\zeta} = 1 \qquad\qquad \sum_{\b\neq \zeta} e_{\zeta\b} = 0 
\end{align}
enforcing that in-degree of $\iota$ and out-degree of $\zeta$ are 0, 
while out-degree of $\iota$ and in-degree of $\zeta$ are 1.
Thus atoms split into three categories: inner nodes (with in- and out-degree equal 1), 
end nodes ($\iota$ and $\zeta$) and non-nodes (with in- and out-degree equal 0).
Therefore, the graph $G$ defined by a finitary solution consists of a directed path from $\iota$ to $\zeta$ 
plus a number of vertex disjoint directed cycles.
The path will be used below to encode a run of $M$:
each edge, intuitively speaking, will be assigned a configuration of $M$, while each inner node will be assigned an instruction of $M$.

The system $S$ has also unknowns $t_{i \alpha}$ indexed by instructions $i\in I$ of $M$ and atoms
$\a\in\A$, and the following equations:
\begin{align}\label{eq:instr constraint}
\sum_{i \in I} t_{i \a} &= \sum_{\b\neq \a} e_{\a\b} \qquad\qquad (\a\in\A\setminus\set{\iota, \zeta}).
\end{align}
Therefore in every finitary solution, for each inner node $\a$ of the above-defined graph $G$,
there is exactly one instruction $i\in I$ such that $t_{i \a}$ equals 1
(intuitively, this instruction $i$ is \emph{assigned} to node $\a$), 
and $t_{i\a}$ equals 0 for all other instructions.
(This applies to \emph{all} inner nodes of $G$, both those on the path as well as those on cycles.)
For non-nodes $\a$, all $t_{i\a}$ are necessarily equal 0.
Note that the values of unknowns $t_{i\iota}$ and $t_{i\zeta}$ are unrestricted, as they are irrelevant.

Finally, the system $S$ contains unknowns $c_{\a\b\c k}$ indexed by $\a\b\c\in\otu \A 3$ and 
$k\in\setto d$.
The following inequalities:
\begin{align}\label{eq:cleqe}
c_{\a\b\c k} & \leq e_{\a\b} \qquad\qquad (\a\b\c\in\otu \A 3, k \in\set{1\ldots d})
\end{align}
enforce that, whatever atom $\c$ is, the value of unknown $c_{\a\b\c k}$ may be 0 or 1 when $\a\b$ is an edge
(i.e., when $e_{\a\b}=1$), but $c_{\a\b\c k}$ is forcedly 0 when $\a\b$ is a non-edge (i.e., when $e_{\a\b}=0$).
The underlying intuition is that for each $k\in\setto d$, we represent the $k$th coordinate of the configuration \emph{assigned} to the edge $\a\b$ by the 
(necessarily finite) sum
\begin{align} \label{eq:encconf}
\sum_{\c\notin\set{\a, \b}} c_{\a\b\c k}.
\end{align}
(In particular, configurations assigned to non-edges are necessarily zero on all coordinates.)
In agreement with this intuition,
we add to $S$ the requirement that the configuration assigned to the edge 
outgoing from $\iota$ is the source $c_0$,
and the configuration assigned to the edge incoming to $\zeta$ is the target $c_f$:
\begin{align*}
\sum_{\b,\c\neq \iota} c_{\iota\b\c k}  = c_0(k)  \qquad\qquad  
\sum_{\b,\c\neq \zeta} c_{\b\zeta\c k}  = c_f(k)  \qquad\qquad (k \in \set{1,\ldots, d}).
\end{align*}
%
Furthermore, in order to enforce correctness of encoding of a run of $M$,
we add to $S$ equations that relate, intuitively speaking, two consecutive configurations.
Recall that, due to \eqref{eq:degree constraint one}--\eqref{eq:degree constraint two} 
and \eqref{eq:cleqe}, for every $\a\in\A\setminus\set{\iota,\zeta}$, 
unknowns $c_{\b\a\c k}$ may be positive for at most one 
$\b\in\A$; likewise unknowns $c_{\a\b \c k}$.
We add to $S$ the following equations:
\begin{align}\label{eq:update constraint}
\sum_{\b,\c\neq \a} c_{\b\a\c k} \ + \sum_{i \in \updinstr(k)} i(k) \cdot t_{i \a}
\ = \sum_{\b,\c\neq \a} c_{\a\b\c k} 
\qquad\qquad
(\a\in\A\setminus\set{\iota, \zeta}, k \in \set{1\ldots d}).
\end{align}
%
These equalities say that for every inner node or non-node $\a$ (i.e., every atom except for the end nodes $\iota$ and $\zeta$),
on every coordinate $k$, the configuration incoming to $\a$ differs from the configuration
outgoing from $\a$ exactly by the sum
\[
\sum_{i \in \updinstr(k)} i(k) \cdot t_{i \a}
\]
ranging over those instructions $i$ of $M$ that update counter $k$.
Remembering that for each $\a$ there is at most one instruction $i$ satisfying $t_{i \a} \neq 0$,
we get that the configurations differ on coordinate $k$ by exactly $i(k)$ 
(if $i$ updates counter $k$) or the configuration are equal on coordinate $k$
(if $i$ zero-tests counter $k$, or there is no instruction $i$ such that $t_{i \a} \neq 0$).

In order to deal with zero tests, we add to $S$ not just the inequalities~\eqref{eq:cleqe},
but the following strengthening thereof:
\begin{align}\label{eq:zero test}
c_{\a\b\c k} + \!\!\!\! \sum_{i\in\zeroinstr(k)} \!\!\! t_{i\a} & \leq e_{\a\b} 
\qquad\qquad (\a\b\c\in\otu \A 3, k \in\set{1\ldots d}).
\end{align}
In consequence, for every edge $\a\b$, if the instruction $i$ assigned to $\a$ updates counter $k$,
\eqref{eq:zero test} does not restrict further the  $k$th coordinate of the configuration assigned to $\a\b$.
But if the instruction $i$ assigned to $\a$ zero-tests counter $k$, the sum
\[
\sum_{i\in\zeroinstr(k)} t_{i\a}
\]
equals $1$ and therefore the $k$th coordinate of the configuration assigned to $\a\b$, 
encoded by~\eqref{eq:encconf}, is necessarily 0
(the same applies also to the configuration incoming to $\a$, due to inequalities \eqref{eq:update constraint} below). 
The above considerations apply to \emph{all} edges of $G$, both those on the path as well as those on cycles.
As a further consequence, for a  non-edge $\a\b$, the configuration assigned at $\a\b$, encoded by~\eqref{eq:encconf}, 
is necessarily the zero configuration.

The construction of $S$ is thus completed, and it remains to argue towards its correctness:
\begin{lemma} \label{lem:undecid-corr}
$M$ admits a run from $c_0$ to $c_f$ if and only if 
$S$ has a finitary nonnegative integer solution.
\end{lemma}
\begin{proof}
For the `if' direction, given a finitary nonnegative integer solution of $S$, we consider
the graph $G$ determined by values of unknowns $e_{\a\b}$, as discussed in the course of construction, 
consisting of inner nodes and two end nodes,
and having the form of a finite directed path plus (possibly) a number of directed cycles. 
By the construction of $S$, each edge of $G$ has assigned a configuration of $M$,
and each inner node has assigned an instruction of $M$, so that the configuration on the edge outgoing
from an inner node is exactly the result of executing its instruction on the configuration assigned to the incoming edge.
(As above, this applies to \emph{all} inner nodes and edges of $G$, both those on the path as well as those on cycles.)
Ignoring the cycles of $G$, we conclude that the sequence of configurations and
instructions along the path of $G$ is a run of $M$ from $c_0$ to $c_f$.

For the `only if' direction, given a run of $M$ from $c_0$ to $c_n$ as in~\eqref{eq:run},
one constructs a solution of $S$ in the form of a sole path involving end nodes $\a_0 = \iota, \a_{n+1} = \zeta$,
$n$ inner nodes $\a_1, \ldots, \a_n$, and $n+1$ edges $\a_j \a_{j+1}$. 
Thus unknowns $e_{\a_j, \a_{j+1}}$ are equal 1.
The values of unknowns $t_{i \a_j}$ are determined by instructions $i_j$ used in the run,
and the values of the unknowns $c_{\a_j \a_{j+1} \c k}$, for sufficiently many fresh atoms $\c$, are determined
by configurations $c_j$. 
All other unknowns are equal 0.
\end{proof}

\begin{slremark}
The proof does not adapt to \finnonnegsolvname$(\Z)$.
Indeed, the standard way of transforming inequalities into equations involves 
adding an infinite set of additional unknowns, that might be all non-zero.
\end{slremark}


\section{Conclusions} \label{sec:conc}

As two main contributions, we show two contrasting results:
decidability of orbit-finite linear programming, and 
undecidability of orbit-finite integer linear programming.
For decidability, we invent a novel concept of setwise-$T$-orbit, and provide a reduction
to a finite but polynomially-parametrised linear programming.
In addition to the decidability of the latter problem,
we show that it can be solved in \EXPTIME, and even in \PTIME for every fixed atom dimension. 
We thus match, in case of fixed atom dimension, the complexity of classical linear programming.

We consider non-strict inequalities for presentation only, 
and our decision procedures
may be straightforwardly adapted to to mixed systems of strict and non-strict inequalities.

We leave a number of intriguing open questions, all of them except the last one referring to
linear programming:

\begin{question}
In this paper we only consider finitely supported solutions.
We do not know the decidability status of 
linear programming when this restriction is dropped
(like in \cite{KKOT15}).
It is decidable for finitary inequalities, where existence of a solution implies existence of an equivariant one \cite{szymon-oral}.
\end{question}

\begin{question}
We exclusively consider equality atoms, and
extension to richer structures seems highly non-trivial.
In the important case of \emph{ordered atoms},
we are currently only able to prove decidability of 
\solvname$(\F)$, for any commutative ring $\F$.
\end{question}

\begin{question}
It is very natural to ask if the classical duality of linear programs extends to the orbit-finite setting.
According to our initial observations this is indeed the case, under the restriction that either 
(v) vertical vectors of a matrix and the target vector are finitary, or
(h) horizontal vectors of a matrix and objective function are finitary. 
Whenever the primary program satisfies one of the conditions (v), (h), the dual one satisfies the other one.
\end{question}

\begin{question}
Solution sets of orbit-finite systems are not always finitely generated.
Therefore an interesting question arises if one can compute a representation of solution sets
that would enable testing for equality or inclusion of such sets?
For instance, the solution set of the system of inequalities \,
$\sum_{\a \in \A\setminus\set{\b}} \a  \ \geq \ \b$ \ \ $(\b \in \A)$, 
in matrix form
\begin{align*} 
& 
\begin{bmatrix}
\ -1 & 1 & 1 & \cdots \ \\
\ 1 & -1 & 1 & \cdots \ \\
\ 1 & 1 & -1 & \cdots \ \\
 \vdots     & \vdots & \vdots & \ddots \   
\end{bmatrix}
\cdot
\ \vr x 
\quad
\geq
\quad
\begin{bmatrix}
\, 0 \, \\
\, 0 \, \\
\, 0 \, \\
 \vdots 
\end{bmatrix}
\end{align*}
is not equal to the cone generated by (=non-negative linear combinations of) an orbit-finite set.
%
%
\end{question}

\begin{question}
We would be happy to know if our general \EXPTIME upper complexity bound is tight.
\end{question}

\begin{question}
Concerning integer linear programming,
an intriguing research task is to identify the decidability borderline.
For instance, we suspect decidability in case when all inequalities are finitary.
The reduction proving undecidability produces a system of atom dimension $3$, and it is unclear if the dimension can be lowered
to $2$.
In case of atom dimension 1 we suspect decidability (along the lines of \cite{HLT17}).
\end{question}


\begin{acks}
We are grateful to Miko{\l}aj Boja{\'n}czyk, Asia Fijalkow, Lorenzo Clemente and Damian Niwi{\'n}ski 
for fruitful discussions.

We acknowledge partial support of 
NCN Opus grant 2019/35/B/ST6/02322 (Arka Ghosh and S{\l}awomir Lasota),
of NCN Preludium grant 2022/45/N/ST6/03242 (Arka Ghosh),
and of
ERC Starting grant INFSYS, agreement no.~950398 (Piotr Hofman and S{\l}awomir Lasota).
\end{acks}

\bibliographystyle{ACM-Reference-Format}
\bibliography{acmart,bib}


\appendix

\section{Missing proofs}


We start by introducing notation 
useful in proving Theorems \ref{thm:equiv-problems} and \ref{thm:equiv-max}.
For subsets $P\subseteq \GLin B$ and $\F\subseteq \R$, we define $\Span {\F} P\subseteq \GLin B$ as the set of 
all linear $\F$-combinations of vectors from $P$:
\begin{align*} 
\begin{aligned}
 \Span {\F} P \ = \ \  \setof{q_1 \cdot \vr p_1 + \ldots + q_k \cdot \vr p_k}{k\geq 0,
\ q_1,\ldots,q_k\in \F, \ \vr p_1, \ldots, \vr p_k\in P}.
\end{aligned}
\end{align*}
Recall that given a matrix $\vr A \in \GLin {B\times C}$ with rows $B$ and columns $C$, we can define a partial operation of 
multiplication of $\vr A$ by a vector $\vr v\in\GLin C$ in an expected way:
\[
(\mult {\vr A}{\vr v})(b) = \innerprod {\vr A(b, \_)} {\vr v}
\]
for every $b\in B$.
The result $\mult {\vr A} {\vr v} \in \GLin B$ is well-defined if $\innerprod {\vr A(b, \_)}{\vr v}$ is well-defined for all $b\in B$.
For $c\in C$ we denote by $\vr A(\_, c) \in \GLin B$ the corresponding (column) vector.
The multiplication $\mult {\vr A}{\vr v}$ can be also seen as an \emph{orbit-finite} linear combination of 
column vectors $\vr A(\_, c)$, for $c\in C$, with coefficients given by $\vr v$. 
This allows us to define the \emph{span} of $\vr A$ seen as a $C$-indexed orbit-finite set of 
vectors $\vr A(\_, c) \in \GLin B$:
\begin{align*} 
\GSpan {\F} {\vr A}  :=  \setof{\mult{\vr A}{\vr v}}{\vr v : C\tofs \F, \ \mult{\vr A}{\vr v} \text{ well-def.}}.
\end{align*}
Therefore, a system of inequalities $(\vr A, \vr t)$ has a solution if 
$\GSpan {\F} {\vr A}$ contains some vector $\vr u \geq \vr t$.
When $\vr v$ is finitary, well-definedness is vacuous, and we may define:
\[
\Span {\F} {\vr A}  :=  \setof{\mult{\vr A}{\vr v}}{\vr v : C\tofin \F}  =  \Span {\F} P
\]
for $P = \setof{\vr A(\_, c)}{c\in C}$ the set of column vectors of $\vr A$.
Therefore, a system of inequalities $(\vr A, \vr t)$ has a finitary solution if 
$\Span {\F} {\vr A}$ contains some vector $\vr u \geq \vr t$.

\subsection{Proof of Theorems \ref{thm:equiv-problems} and \ref{thm:equiv-max} 
(Section \ref{sec:problems})}
\label{sec:problems-proofs-one}

Recall that we consider supremum of a maximisation problem to be $-\infty$ if the constraints in the problem are infeasible.
Therefore proving that two maximisation problems have the same supremum also proves that the underlying systems of inequalities are equisolvable.
In consequence, Theorem \ref{thm:equiv-max} implies \ref{thm:equiv-problems}, and hence
we concentrate
in the sequel on proving the former one.

The proof of mutual
reductions between \ineqmaxname$(\F)$ and \nonnegmaxname$(\F)$
amounts to lifting of standard arguments from finite to orbit-finite systems, 
and checking that all constructed objects are finitely supported. 
We include the reductions here mostly in order to get acquainted with orbit-finite systems.
One of the remaining two reductions builds on results of \cite{GHL22}.

\para{Reduction of \ineqmaxname$(\F)$ to \nonnegmaxname$(\F)$}
Consider an instance $(\vr{A},\vr{t},\vr{s})$ of \ineqmaxname$(\F)$, supported by $S$,
where $\vr A : B \times C \tofs \F$, $\vr t : B\tofs \F$ and $\vr{s} : C\tofs \F$.
We construct an instance $(\vr A', \vr t, \vr{s}')$ of \nonnegmaxname$(\F)$,
with the same target vector $\vr t$, and
$
\vr A' : B \times (C \uplus C \uplus B) \tofs \F,
$
$\vr{s}' : (C \uplus C \uplus B) \tofs \F$,
such that
\begin{align*} 
\textnormal{supremum}(\vr A', \vr t, \vr{s}')
=
\textnormal{supremum}(\vr A, \vr t, \vr{s}).
\end{align*}
In the new system, we double each variable $x$ into $x_+$ and $x_-$, and we add a fresh variable per each equation.
The matrix $\vr A'$ of the new system is a composition of $\vr A$, $-\vr A$, and the diagonal matrix 
$B \times B\tofs \F$ with $-1$ in the diagonal:

\vspace{-2mm}
\begin{align*}
\vr A' \ = \ 
\left[
\begin{matrix}
  & &  \phantom{1} \\  & \ \ \ \vr A & \phantom{\ddots}    \\     &  &    \phantom{1}
\end{matrix}
\right\rvert
\begin{matrix}
  & &   \\  & \ \ - \vr A &  \phantom{\ddots}  \\     &  &   
\end{matrix}
\left\rvert
\begin{matrix}
\; -1 & &   \\  & \ddots &    \\     &  & -1 \;  
\end{matrix}
\right]
\end{align*}
%

\noindent
Similarly, $\vr{s}'$ is defined as the composition of $\vr{s}$, $-\vr{s}$ and the zero vector $B \tofs \F$:

\vspace{-4mm}
\begin{align*}
\vr s' \ = \ 
\left[
\begin{matrix}
  & \vr s & 
\end{matrix}
\right\rvert
\begin{matrix}
  & - \vr s &   
\end{matrix}
\left\rvert
\begin{matrix}
\; 0 & \cdots & 0 \;  
\end{matrix}
\right]
\end{align*}
$\vr A'$ and $\vr{s}'$ are thus supported by $S$.

Any vector 
$
\vr{x}' :(C \uplus C \uplus B) \tofs \F
$
can be written as
\[
\vr x' \ = \ (\vr{x}_+ \rvert \vr{x}_- \rvert \vr{y}),
\]
where $\vr{x}_+,\vr{x}_- : C \tofs \F$ and $\vr{y} : B \tofs \F$.
If any such non-negative vector $\vr x'$ satisfies the above constructed system of constraints, i.e. if we have
\begin{equation}\label{eq: non-neg max equiv to ineq}
\vr{A}' \cdot (\vr{x}_+ \rvert \vr{x}_- \rvert \vr{y}) = \vr{t},
\end{equation}
then then vector ${\vr{x}_+ - \vr{x}_-}$, supported by $\supp{\vr{x}'}$,
is a solution of $(\vr A, \vr t)$, namely
\[
\vr{A} \cdot (\vr{x}_+ - \vr{x}_-) \geq
\vr{A} \cdot (\vr{x}_+ - \vr{x}_-) - \vr{y} =
\vr{A}' \cdot (\vr{x}_+ \rvert \vr{x}_- \rvert \vr{y}) = \vr{t}.
\]
Furthermore, by the very definition of $\vr s'$ we have
\begin{align} \label{eq:ss}
\vr{s}' \cdot (\vr{x}_+ \rvert \vr{x}_- \rvert \vr{y}) = \vr{s} \cdot (\vr{x}_+ - \vr{x}_-),
\end{align}
which implies
$
\textnormal{supremum}(\vr A', \vr t, \vr{s}')
\leq
\textnormal{supremum}(\vr A, \vr t, \vr{s}).
$

In the opposite direction, given a finitely supported vector $\vr{x}$ such that 
$\vr{A} \cdot \vr{x} \geq \vr{t}$, we define a non-negative vector 
$\vr x' = (\vr{x}_+ \rvert \vr{x}_- \rvert \vr{y})$ supported by $\supp{\vr{x}}\cup S$ as follows:

\begin{minipage}{0.35\linewidth}
\begin{align*}
\vr{x}_+(c) & =
\begin{cases}
\vr{x}(c) & \text{ if }\vr{x}(c) \geq 0, \\
0         & \text{ otherwise;}
\end{cases} 
\end{align*}
\end{minipage}
\begin{minipage}{0.35\linewidth}
\begin{align*}
\vr{x}_-(c) & =
\begin{cases}
-\vr{x}(c) & \text{ if }\vr{x}(c) < 0, \\
0         & \text{ otherwise;}
\end{cases} 
\end{align*}
\end{minipage}
\begin{minipage}{0.3\linewidth}
\begin{align*}
\vr{y}(c) & = (\vr{A}\cdot \vr{x} - \vr t)(c). \\
\ 
\end{align*}
\end{minipage}
\smallskip

\noindent
Then $\vr{x} = \vr{x}_+ - \vr{x}_-$ and
$
\vr{A}' \cdot (\vr{x}_+ \rvert \vr{x}_- \rvert \vr{y}) =
\vr{A} \cdot \vr{x} - \vr{y} = \vr{t}.
$
The equality \eqref{eq:ss} holds again,
which implies 
$
\textnormal{supremum}(\vr A, \vr t, \vr{s})
\leq
\textnormal{supremum}(\vr A', \vr t, \vr{s}').
$

\para{Reduction of \nonnegmaxname$(\F)$ to \ineqmaxname$(\F)$}
For any orbit-finite system of linear equations supported by $S$:
\begin{align*}
\vr{A} \cdot \vr{x} = \vr{t},
\end{align*}
its nonnegative solutions are exactly solutions of the following system of linear inequalities, also supported by $S$:
\begin{align*}
\vr{A} \cdot \vr{x} \geq \vr{t} \qquad\qquad 
\vr{A} \cdot \vr{x} \leq \vr{t} \qquad\qquad 
\vr{x} \geq 0.
\end{align*}
This implies an easy reduction from \nonnegmaxname$(\F)$ to \ineqmaxname$(\F)$.

\begin{slremark}
The above two reductions preserve row-finiteness, i.e., transform a system of finite equations (inequalities) to a system of finite inequalities (equations), or vice versa.
\end{slremark}

\para{Reduction of \finineqmaxname$(\F)$ to \ineqmaxname$(\F)$}
Consider an instance $(\vr{A},\vr{t},\vr{s})$ of \finineqmaxname$(\F)$ supported by $S$,
where $\vr A : B\times C\tofs \F$.
We construct an instance $(\vr{A}',\vr{t}',\vr{s}')$ of \ineqmaxname$(\F)$ as follows.
The new system of inequalities
$\vr{A}' \cdot \vr{x}' \geq \vr{t}'$ is obtained by
extending the column index $C$ by one additional variable $y$ and extending the system by one inequality:

\vspace{-2mm}
\begin{align*}
\vr A' \ = \ 
\left[
\begin{matrix}
  & & \phantom{0}  \\  & \ \ \vr A \ \ \phantom{\vdots} &     \\     &  & \phantom{0} \\ \hline
  1 & \cdots & 1
\end{matrix}
\right\rvert
\hspace{-1.25mm}
\left\rvert
\begin{matrix}
  0 \\     \vdots \\   0 \\       \hline -1
\end{matrix}
\right]
\qquad\qquad
\vr t' \ = \ \begin{bmatrix}
\phantom{0} \\ \ \vr t \phantom{\vdots} \\ \phantom{0} \\ \hline 0 
\end{bmatrix}
\end{align*}
and the new objective function $\vr s'$ is defined as expected:

\vspace{-4mm}
\begin{align*}
\vr s' \ = \ 
\left[ 
\begin{matrix}
& \;   \vr s  \; & 
\end{matrix}
\rvert
\; 0 \;  
\right].
\end{align*}
The so constructed instance is supported by $S$, and its solutions have the form $\vr x' = (\vr x, y)$, where
\[
\vr{A} \cdot \vr{x} \geq \vr{t} \qquad\qquad \sum_{c \in C} \vr{x}(c) \geq y.
\]
Any such finitely supported solution is necessarily finitary.
This implies
$
\textnormal{supremum}(\vr{A},\vr{t},\vr{s}) =
\textnormal{supremum}(\vr{A}',\vr{t}',\vr{s}).
$
%

\para{Reduction of \ineqmaxname$(\F)$ to \finineqmaxname$(\F)$}

We rely on the following result of~\cite{GHL22}%
\footnote{The result, as shown in~\cite{GHL22}, holds for any commutative ring $\F$.}:
\begin{claim}[\cite{GHL22} Claim~20]
Let $\F \in\set{\Int,\R}$.
Given an $S$-supported orbit-finite matrix $\vr{M}$ one can effectively construct an $S$-supported orbit-finite matrix $\spread{\vr{M}}$
such that $\GSpan{\F}{\vr{M}} = \Span{\F}{\spread{\vr{M}}}$.
\end{claim}
Consider an instance $(\vr{A},\vr{t},\vr{s})$ of \ineqmaxname$(\F)$, and
apply the above claim to the matrix $\vr M$ (left) 
in order to get the matrix $\spread{\vr{M}}$ (right),
\[
\vr{M} \ = \ 
\left[
\begin{matrix}
            &        &  \\
            & \vr A  &  \\
            &        &  \\
              \hline
\phantom{1} & \vr{s} & \phantom{1}
\end{matrix}
\right]
\qquad\qquad\qquad\qquad
\spread{\vr{M}} \ = \ 
\left[
\begin{matrix}
            &                  &  \\
            & \vr {A}' &  \\
            &                  &  \\
              \hline
\phantom{1} & \vr{s}'  & \phantom{1}
\end{matrix}
\right]
\]
such that
\begin{align} \label{eq:spany}
\GSpan{\F}{\vr{M}} =
\Span{\F}{\spread{\vr{M}}}.
\end{align}
This yields 
an instance  $(\vr{A}',\vr{t},\vr{s}')$  of \finineqmaxname$(\F)$, supported 
by $\supp{\vr{A},\vr{t},\vr{s}}$.

By the equality \eqref{eq:spany},
for every $r \in \R$ we have the following:
there exists a finitely supported vector $\vr{x}$ such that 
$\mult{\vr{A}}  {\vr{x}} \geq \vr{t}$ and $\innerprod{\vr{s}}  {\vr{x}} = r$ 
if and only if
there exists a finitary vector $\vr{x}'$ such that $\mult{\vr{A}'}  {\vr{x}'} \geq \vr{t}$ and 
$\innerprod {\vr{s}'}  {\vr{x}'} = r$.
In consequence,
\[
\text{supremum}(\vr{A},\vr{t},\vr{s}) =
\text{supremum}(\vr{A}',\vr{t},\vr{s}').
\]

Theorem \ref{thm:equiv-max} is thus proved.

\subsection{Proof of Theorem~\ref{thm:finnegsolv-decid} (Section \ref{sec:problems})}
\label{sec:problems-proofs-two}

We consider cases $\F = \R$ and $\F = \Z$ separately. 

\para{Case $\F = \R$}

Decidability of \finnonnegsolvname$(\R)$ follows by a direct reduction of 
\finnonnegsolvname$(\R)$ to \nonnegsolvname$(\R)$
(similar to the reduction of \finineqmaxname$(\F)$ to \ineqmaxname$(\F)$)
and Theorem~\ref{thm:lp-decid}.

\para{Case $\F = \Z$}

Decidability of \finnonnegsolvname$(\Int)$ follows by results of \cite{GHL22} and \cite{HR21}.

Let $(\vr{A},\vr{t})$ be an instance of \mbox{\finnonnegsolvname$(\Int)$}, where
$\vr A : B\times C \tofs \Z$, and
consider the set of column vectors 
\[
P = \setof{\vr{A}(\_,c)}{c \in C} \subseteq \GLin B
\]
of $\vr{A}$.
Then the system of equations $\vr{A} \cdot \vr{x} = \vr{t}$ has a finite 
non-negative integer solution if and only if 
\begin{align}\label{eq:tspan}
\vr{t} \in \Span{\Nat}{P}.
\end{align}
We rely on Theorem 3.3 of \cite{GHL22} which says that $\GLin B$ has an orbit-finite basis.
Let $\widehat{B} \subseteq \GLin B$ be such a basis.
This implies that there exists a linear isomorphism 
$
\varphi : \GLin B \to\Lin{\widehat{B}}.
$
In consequence, \eqref{eq:tspan} is equivalent to
\begin{align} \label{eq:tspan2}
\varphi(\vr{t}) \in \Span{\Nat}{\varphi(P)}.
\end{align}
By Remark 11.16 of \cite{HR21} we can compute a finite set of vectors 
$\set{\vr{t}'_1,\dots,\vr{t}'_k}\subseteq \Lin {\widehat B}$ and an orbit-finite subset 
$P'\subseteq\varphi(P)$ such 
that \eqref{eq:tspan2}
holds if and only if 
\begin{align} \label{eq:tspan3}
\vr{t}'_i \in \Span{\Int}{P'}
\end{align} 
for some $i\in\setto k$.
The question \eqref{eq:tspan3} is nothing but finitary integer solvability of an orbit-finite system of equations,
which is decidable using Theorem 6.1 of \cite{GHL22}.

%
%
%

\subsection{Proof of Theorem  \ref{thm:polydecid} (Section \ref{sec:poly})}
\label{sec:poly-proofs}

We show decidability of \polyineqsolvname 
by 
encoding the problem into real arithmetic, i.e., first-order theory of $(\R, +, \cdot, 0, 1, \leq)$.
We say that a real arithmetic formula $\varphi(x_1, \ldots, x_k)$ with free variables $x_1, \ldots, x_k$
\emph{defines} the set of all valuations $\set{x_1, \ldots, x_k}\to\R$ satisfying it. 
When the order of free variables is fixed, we naturally identify the set defined by $\varphi$ with a subset of $\R^k$.
\begin{claim} \label{claim:1var}
Every real arithmetic formula $\varphi(x)$ with one free variable, 
defines
a finite union of (possibly infinite) disjoint intervals.
\end{claim}
\begin{proof}
By quantifier elimination~\cite{Tar51}, the formula $\varphi(x)$ is equivalent to a 
quantifier-free formula $\overline\varphi(x)$ with constants, namely $\varphi(x)$ and $\overline\varphi(x)$ define the same set.
Therefore $\overline\varphi(x)$ is a Boolean combination of inequalities
$
p(x) \geq 0,
$
for univariate polynomials $p\in\polyring \R x$, and validity of $\overline\varphi(x)$
depends only on the sign of $p(x)$, for (finitely many) polynomials that appear in $\overline\varphi(x)$.
This implies the claim.
\end{proof}

Consider a fixed system $P$ of polynomially-parametrised inequalities over unknowns $x_1, \ldots, x_k$, and
let $n$ range over \emph{reals}, not just over nonnegative integers.
For each $n\in\R$, we  get the system $P(n)$ of linear inequalities with \emph{real} coefficients.  
Let 
\begin{align} \label{eq:sigmaP}
\sigma_P(n, x_1,\ldots, x_k)
\end{align}
 be the conjunction of inequalities in $P$, each of the form~\eqref{eq:polyeq};
it is thus a quantifier-free real arithmetic formula which says
that a  tuple $\vr x = x_1, \ldots, x_k$ is a solution of  $P(n)$. 
The existential real arithmetic formula $\psi(n) \ \equiv \ \prettyexists{\vr x}{\sigma_P(n, \vr x) }$ 
with one free variable $n$,
says that $P(n)$ has a real solution.
Thus 
\polyineqsolvname has positive answer exactly when $\psi(n)$ is true for some $n\in\Nat$.

Evaluating real arithmetic formulas of fixed quantifier alternation depth is doable in \EXPTIME~\cite{Ben-OrKR86}, 
\cite[Theorem 14.16]{ra-book}.
In order to decide \polyineqsolvname, the algorithm evaluates the closed formula
\begin{align*} 
\prettyexists{\tilde n}{\prettyforall{n}{n >\tilde n {\implies} \psi(n)}}
%
\end{align*}
and answers positively if the formula is true.
Otherwise, we know that the set $D$ defined by $\psi$, being a finite union of intervals (cf.~Claim~\ref{claim:1var}), 
is bounded from above.
The algorithm computes an integer upper bound $m_0$ of $D$, by evaluating closed existential formulas
\[
\varphi_m \ \equiv \ \prettyexists n  {n > m \ \land \ \psi(n)},
\]
for increasing nonnegative integer constants $m = 0, 1, \ldots$,
until $\varphi_m$ eventually evaluates to false.
Finally, the algorithm evaluates the formula $\psi(m)$ for all nonnegative integers $m$ between $0$ and $m_0$,
and answers positively if $\psi(m)$ is true for some such $m$; otherwise the algorithm answers negatively.

\subsection{Proofs of Lemmas \ref{lem:wlog} and \ref{lem:wlogmax} (Sections \ref{sec:lp-decid-proof} and \ref{sec:max})}
\label{sec:lp-decid-proof-proofs}

%
We sketch the proofs only, as they amount to a slightly tedious but entirely standard exercise in sets with atoms.

Consider an instance $(\vr{A},\vr{t},\vr{s})$ of the maximisation problem \finineqmaxname$(\R)$.
Let $S = \supp{\vr{A},\vr{t},\vr{s}}$, and let
$\vr A : B\times C\tofs \Z$.
Thus the row and column index sets $B$ and $C$ are necessarily supported by $S$.
We want to effectively transform the instance into another one $(\spread{\vr{A}},\spread{\vr{t}},\spread{\vr{s}})$, 
where the row and column index sets are 
disjoint unions of sets of the form $\otu \A \ell$ (non-repeating tuples of atoms of a fixed length), as in
\eqref{eq:eqorbits} in Section \ref{sec:prelimproof}.
%
Moreover, the transformation should preserve the supremum:
\begin{align} \label{eq:supsupprove}
\text{supremum}(\vr{A},\vr{t},\vr{s}) =
\text{supremum}(\spread{\vr{A}},\spread{\vr{t}},\spread{\vr{s}}).
\end{align}

Recall that we consider supremum of a maximisation problem to be $-\infty$ if the constraints are infeasible.
Therefore proving that two maximisation problems have the same supremum also proves that the underlying systems of inequalities are equisolvable.


%
%

We proceed in two steps.
First we show that the row and column index sets $B$ and $C$
may be assumed to be disjoint unions of sets of the form $\otu {(\A\setminus S)} \ell$.
As mentioned in Section \ref{sec:results},
$B$ and $C$ are assumed to be given as finite union of $S$-orbits of the form $(\A \setminus S)^{(n)}/_G$ where $n \in \Nat$ and $G$ is a subgroup of $S_n$, the group of all permutations of the set $\set{1,\dots,n}$.
Consider the partition of $B$ and $C$ into $S$-orbits:
\[
B = B_1 \uplus \dots \uplus B_k
\qquad\qquad
C = C_1 \uplus \dots \uplus C_{\ell},
\]
where
\[
B_i = (\A \setminus S)^{(p_i)}/_{G_i}
\qquad \text{ and }
\qquad
C_j = (\A \setminus S)^{(q_j)}/_{H_j}
\]
for some $p_i,q_j \in \Nat$ and subgroups $G_i$ and $H_j$ of respectively $S_{p_i}$ and $S_{p_j}$.
Let $f_i$ and $g_j$ be the quotient maps
\[
f_i : (\A \setminus S)^{(p_i)} \to B_i \qquad
g_j : (\A \setminus S)^{(q_j)} \to C_j.
\]
Notice that for every $i,j$ and $x \in B_i$ and $y \in C_j$
\begin{equation}\label{eq:const size preimage}
|f_i^{-1}(x)| = |G_i| \qquad\qquad
|g_j^{-1}(y)| = |H_j|.
\end{equation}
%
%
%
%
%
We put:
\begin{align} \label{eq:BCdef}
\begin{aligned}
B' \  = \ (\A \setminus S)^{(p_1)} \uplus \dots \uplus (\A \setminus S)^{(p_k)} \qquad\qquad
C' \  = \ (\A \setminus S)^{(q_1)} \uplus \dots \uplus (\A \setminus S)^{(q_{\ell})}
\end{aligned}
\end{align}
and define maps $f : B' \to B$ and $g: C' \to C$ by disjoint unions of $f_1, \ldots, f_k$ and
$g_1, \ldots, g_\ell$, respectively:
\[
f = f_1 \uplus \dots \uplus f_k
\qquad\qquad
g = g_1 \uplus \dots \uplus g_{\ell}.
\]
Both the maps are surjective.
We write $(f, g): B'\times C' \to B\times C$ for the product of the two maps.
Finally, we
define a matrix $\vr{A}' : (B' \times C') \tofs \Int$ and vectors $\vr{t}' \in \GLin {B'}$ and
$\vr{s}' \in \GLin{C}'$ by pre-composing with the above-defined maps:
\begin{align} \label{eq:precom}
\vr A' = \vr A \circ (f, g) 
\qquad\qquad
\vr{t}'      = \vr{t} \circ f  \qquad\qquad
\vr{s}'      = \vr{s} \circ g.
\end{align}
%
%
\begin{lemma}\label{lem:canon S supp}
$
\textnormal{supremum}(\vr{A},\vr{t},\vr{s}) =
\textnormal{supremum}(\vr{A}',\vr{t}',\vr{s}').
$
\end{lemma}
\begin{proof}
Define two functions $F : \GLin{C} \to \GLin{C'}$ and $G : \GLin{C'} \to \GLin{C}$ as follows:
\begin{align*}
F(\vr{x}) \  : \  c' \mapsto\  \frac {\vr{x}(g(c'))} {|H_i|}, \text{ where }g(c') \in C_i 
\qquad\qquad
G(\vr{x}') \  : \ c\ \mapsto \sum_{g(c') = c} \vr{x}'(c').
\end{align*}
%
%
%
Both $F$ and $G$ are supported by $S$.
By the very definition of $F$ and $G$, together with \eqref{eq:const size preimage}, we deduce
the following two facts, assuming either
$\vr x' = F(\vr x)$ or $\vr x = G(\vr x')$, where $\vr x\in \GLin C$ and $\vr x' \in \GLin{C'}$.
First,
the value of $\mult {\vr{A}}  {\vr{x}}$ is well-defined if and only if
the value of $\mult{\vr{A}'} {\vr{x}'}$ is so, and in such case
\[
\mult {\vr{A}}  {\vr{x}} \geq \vr{t} \iff
\mult{\vr{A}'} {\vr{x}'} \geq \vr{t}'.
\]
Second,
the value of $\innerprod {\vr{s}}  {\vr{x}}$ is well defined if and only if 
the value of $\innerprod {\vr{s}'} {\vr{x}'}$ is so, and in such case
$
\innerprod {\vr{s}}  {\vr{x}} = \innerprod {\vr{s}'} {\vr{x}'}.
$
The two facts prove the lemma.
%
%
%
%
%
\end{proof}

The instance $(\vr{A}',\vr{t}',\vr{s}')$ is supported by $S$.


As the second step we transform the instance $(\vr{A}',\vr{t}',\vr{s}')$ further so that
the row and column index sets $B$ and $C$ are disjoint unions of sets of the form $\otu \A \ell$.
Let $h : \A \to \A \setminus S$ be an arbitrarily chosen bijection.
Since atoms from $S$ do not appear in tuples belonging to $B'$ or $C'$, the map $h$
induces two further bijective maps 
\[
f: \spread{B} \to B' \qquad\qquad g : \spread{C} \to C',
\]
where
\[
\spread{B} = \A^{(p_1)} \uplus \dots \uplus \A^{(p_k)}
\qquad\qquad
\spread{C} = \A^{(q_1)} \uplus \dots \uplus \A^{(q_{\ell})}
\]
(cf.~\eqref{eq:BCdef}).
%
%
We define a matrix $\spread{\vr{A}} : \spread B \times \spread C \tofs \Z$ 
and two vectors $\spread{\vr{t}} : \spread B \tofs \Z$ and 
$\spread{\vr{s}} : \spread C \tofs \Z$ by pre-composing with the two above-defined maps,
similarly as in \eqref{eq:precom}:
\[
\spread {\vr A} = \vr A' \circ (f,g) \qquad\qquad
\spread {\vr t} = \vr t' \circ f \qquad\qquad
\spread {\vr s} = \vr s' \circ g.
\]
Knowing that $\vr A'$, $\vr t'$ and $\vr s'$ are all supported by $S$, we deduce that
the so defined instance $(\spread{\vr{A}},\spread{\vr{t}}, \spread{\vr s})$ is equivariant and independent from
the choice of the bijection $h : \A \to \A\setminus S$.
The size blowup is exponential only in atom dimension of $(\vr A, \vr t, \vr s)$, and hence polynomial
when atom dimension if fixed.

\begin{lemma}\label{lem:equivariant sys}
$
\textnormal{supremum}(\vr{A}',\vr{t}',\vr{s}') =
\textnormal{supremum}(\spread{\vr{A}},\spread{\vr{t}},\spread{\vr{s}}).
$
\end{lemma}
\begin{proof}
Similarly as before,
assuming $\vr x' = g(\spread{\vr x})$ for some vectors $\vr x' \in \GLin{C'}$ and $\spread{\vr x} \in \GLin{\spread C}$,
we deduce the following two facts.
First,
the value of $\mult {\vr{A}'}  {\vr{x}'}$ is well-defined if and only if
the value of $\mult{\spread{\vr{A}}} {\spread{\vr{x}}}$ is so, and in such case
\[
\mult {\vr{A}'}  {\vr{x}'} \geq \vr{t'} \iff
\mult{\spread{\vr{A}}} \spread{{\vr{x}}} \geq \spread{\vr{t}}.
\]
Second,
the value of $\innerprod {\vr{s}'}  {\vr{x}'}$ is well defined if and only if 
the value of $\innerprod {\spread{\vr{s}}} {\spread{\vr{x}}}$ is so, and in such case
$
\innerprod {\vr{s}'}  {\vr{x}'} = \innerprod {\spread{\vr{s}}} {\spread{\vr{x}}}.
$
The two facts prove the lemma.
%
%
%
%
%
\end{proof}

The last two lemmas prove Lemmas \ref{lem:wlog} and \ref{lem:wlogmax}.
%

\end{document}